\documentclass[a4paper,UKenglish,cleveref, autoref, thm-restate]{lipics-v2021}

\pdfoutput=1 
\hideLIPIcs  

\nolinenumbers
\usepackage{amssymb}
\usepackage{amsthm}

\usepackage{amsfonts}
\usepackage{stmaryrd}
\usepackage{xcolor, soul}

\usepackage{graphicx}
\usepackage{bbding}
\usepackage{pifont}
\usepackage{paralist}
\usepackage{enumitem}
\usepackage{lineno}
\usepackage{wrapfig}
\usepackage{xspace}

\usepackage{tikz}
\usetikzlibrary{arrows,positioning,shapes,decorations,automata,backgrounds,petri,fit,calc,shapes.multipart,decorations.text,calc,arrows.meta}

\newcommand{\sema}[1]{{\llbracket}#1{\rrbracket}}
\newcommand{\lcm}{\ensuremath{\textsl{lcm}}}

\newcommand{\mincoloru}[1]{\ensuremath{\textsf{MinColor}(#1)}}
\newcommand{\mincolorof}[1]{\ensuremath{\textsf{MinColor}(#1)}}

\renewcommand{\inf}{\textsl{inf}}

\newcommand{\equivcls}[2]{[#1]_{\sim #2}}
\newcommand{\progequiv}[2]{\approx_{#1}^{#2}}
\newcommand{\prog}[2]{\aut{P}_{#1}^{#2}}
\newcommand{\fdfa}[2]{\aut{F}^{#2}(#1)}

\newcommand{\figtext}[1]{\textsc{#1}}

\newcommand{\fdfasubscript}[1]{\textsc{#1}}
\newcommand{\minimlduo}{\fdfasubscript{m}}
\newcommand{\syntactic}{\fdfasubscript{s}}
\newcommand{\colorful}{\text{\tiny{\SixFlowerPetalDotted}}}
\newcommand{\perstaut}[1]{\mathbb{P}_{#1}}

\newcommand{\infcolorof}[1]{\textsf{Color}(#1)}
\newcommand{\fincolorof}[2]{\textsf{color}_{#1}(#2)}
\newcommand{\finColorof}[2]{\textsf{Color}_{#1}(#2)}

\newcommand{\aut}[1]{\ensuremath{\mathcal{#1}}}

\newcommand{\class}[1]{\mathbb{#1}}
\newcommand{\DM}{\class{DM}}
\newcommand{\DMpk}[1]{\ensuremath{\DM_{{#1}}}^{+}}
\newcommand{\DMnk}[1]{\ensuremath{\DM_{{#1}}^{-}}}
\newcommand{\DMpmk}[1]{\ensuremath{\DM_{{#1}}^{\pm}}}

\newcommand{\picm}[1]{\ensuremath{{|#1|_{\subseteq}^{+}}}}
\newcommand{\nicm}[1]{\ensuremath{{|#1|_{\subseteq}^{-}}}}
\newcommand{\pdm}[1]{|#1|_{\leadsto}^{+}}
\newcommand{\ndm}[1]{|#1|_{\leadsto}^{-}}

\tikzset{
spnode/.style={align=center,
        state,
        circle split},
}
\newcommand{\state}[1]{\ensuremath{q_{#1}}}
\newcommand{\statep}[1]{\ensuremath{{#1}}}
\newcommand{\statepp}[1]{\ensuremath{{#1}}}
\newcommand{\stateppp}[1]{\ensuremath{{#1}}}

\newcommand{\nspstate}[2]{\state{#1}/\ensuremath{#2}}

\newcommand{\neginfty}{\ensuremath{-\infty}}

\newcommand{\autofrc}[1]{\aut{A}[{#1}]}

\newcommand{\autstate}[2]{\aut{#1}(#2)}
\newcommand{\auttoq}[2]{\aut{#1}_{#2]}}
\newcommand{\autfromq}[2]{\aut{#1}_{[#2}}
\newcommand{\autfromtoq}[3]{\aut{#1}_{[#2,#3]}}

\newcommand{\figcls}[1]{\textit{#1}}
\newcommand{\figtxt}[1]{\textcolor{blue}{\textsc{#1}}}

\title{A Robust Measure on FDFAs Following Duo-Normalized Acceptance}

\author{Dana Fisman}{Department of Computer Science, Ben-Gurion University, Israel }{dana@cs.bgu.ac.il}{https://orcid.org/0000-0002-6015-4170}{}

\author{Emmanuel Goldberg}{Department of Computer Science, Ben-Gurion University, Israel}{goldbeem@post.bgu.ac.il}{https://orcid.org/0009-0008-6760-1595}{}

\author{Oded Zimerman}{Department of Computer Science, Ben-Gurion University, Israel}{odedzimerman@gmail.com}{https://orcid.org/0009-0003-4020-2037}{Supported by ISF grant 2507/21}

\authorrunning{D. Fisman, E. Goldberg and O. Zimerman} 

\Copyright{Dana Fisman and Emmanuel Goldberg and Oded Zimerman}

\EventEditors{John Q. Open and Joan R. Access}
\EventNoEds{2}
\EventLongTitle{42nd Conference on Very Important Topics (CVIT 2016)}
\EventShortTitle{CVIT 2016}
\EventAcronym{CVIT}
\EventYear{2016}
\EventDate{December 24--27, 2016}
\EventLocation{Little Whinging, United Kingdom}
\EventLogo{}
\SeriesVolume{42}
\ArticleNo{23}


\begin{document}

\maketitle

\begin{abstract}
Families of DFAs (FDFAs) are a computational model recognizing $\omega$-regular languages. They were introduced in the quest of finding a Myhill-Nerode theorem for $\omega$-regular languages, and obtaining learning algorithms.  
FDFAs have been shown to have good qualities in terms of the resources required for computing Boolean operations on them (complementation, union, and intersection) and answering decision problems (emptiness and equivalence); all can be done in non-deterministic logspace. 
In this paper we study FDFAs with a new type of acceptance condition, \emph{duo-normalization}, that generalizes the traditional \emph{normalization} acceptance type. We show that duo-normalized FDFAs are advantageous to normalized FDFAs in terms of succinctness as they can be exponentially smaller. Fortunately this added succinctness doesn't come at the cost of increasing the complexity of Boolean operations and decision problems --- they can still be preformed in non-deterministic logspace. 

An important measure of the complexity of an $\omega$-regular language, is its position in the Wagner hierarchy. It is based on the inclusion measure of Muller automata and for the common $\omega$-automata there exist algorithms computing their position. We 
develop a similarly robust measure for duo-normalized (and normalized) FDFAs, which we term the \emph{diameter measure}.
We show that the diameter measure corresponds one-to-one to the position on the Wagner hierarchy. We show that computing it for duo-normalized FDFAs is PSPACE-complete, while it can be done in non-deterministic logspace for traditional FDFAs.
\end{abstract}

\section{Introduction}\label{sec:intro}

Regular languages of finite words possess a natural canonical representation --- the unique minimal DFA. 
The essence of the representation lies 
in a right congruence relation for a language $L$ saying that two words $x$ and $y$ are \emph{equivalent}, denoted $x\sim_L y$ if and only if $xz\in L \Longleftrightarrow yz\in L$ for every finite word $z\in\Sigma^*$. 
The famous Myhill-Nerode theorem~\cite{Myhill57,Nerode58} relates the equivalence classes of $\sim_L$ to the set of words reaching a state of the minimal DFA.

For regular languages of infinite words the situation is more complex. First, there is no unique minimal automaton for any of the common $\omega$-automata acceptance conditions (B\"uchi, Muller, Rabin, Streett and parity). Second, one can indeed define two finite words $x$ and $y$ to be \emph{equivalent} with respect to an $\omega$-regular language $L$, denoted $x\sim_L y$ if $xz\in L \Longleftrightarrow yz\in L$ for every infinite word $z\in\Sigma^\omega$. However, there is no one-to-one correspondence between these equivalence classes and a minimal $\omega$-automaton for $L$. Consider for instance the language $L_1$ stipulating that $aab$ occurs infinitely often. The right congruence relation $\sim_{L_1}$ has only one equivalence class, yet clearly an automaton for $L_1$ needs more than one state. 

A quest for a characterization of an $\omega$-regular language $L$, relating equivalence classes of a semantic definition of $L$ to states of an automaton for $L$, has led to the development of families of right concurrences (FORCs)~\cite{MalerS97} and families of DFAs (FDFAs)~\cite{AngluinF16}. Several canonical FDFAs were introduced over the years,
the periodic FDFA~\cite{CalbrixNP93old}, the syntactic FDFA~\cite{MalerS97}, the recurrent FDFA~\cite{AngluinF16}, and the limit FDFA~\cite{LST23}.
All these representations have a one-to-one correspondence between the equivalence classes of the right congruence relations and the states of the respective automata.
This is very satisfying in the sense that they induce a semantic canonical representation, i.e. one that is agnostic to a particular automaton; and this is a beneficial property when it comes to learning~\cite{AngluinF16,LiSTCX19}.
FDFAs have additional good qualities --- 
computing Boolean operations on them (complementation, union, and intersection) and answering decision problems (emptiness and equivalence) can all be done cheaply, in non-deterministic logarithmic space~\cite{AngluinBF18}. 

Loosely speaking, an FDFA is composed of a \emph{leading automaton} $\aut{Q}$ and a family of \emph{progress DFAs} $\{\aut{P}_q\}$, one for each state $q$ of $\aut{Q}$. FDFAs consider only ultimately periodic words, i.e. words of the form $u(v)^\omega$ for $u\in\Sigma^*$ and $v\in\Sigma^+$. Since two $\omega$-regular languages recognize the same language if and only if they agree on the set of ultimately periodic words~\cite{Buchi62,CalbrixNP93old}, this is not really a limitation. Acceptance of an ultimately periodic word $(u,v)$ representing $uv^\omega$ is determined by first \emph{normalizing} 
the word wrt to the leading automaton --- this means considering a decomposition $(uv^i,v^j)$ such that $v^j$ loops on the state of the leading automaton reached by $uv^i$ and then checking acceptance of $v^j$ in the respective progress DFA. This normalization was introduced as it leads to an exponential save in the number of states~\cite{AngluinF16}.
In this paper we consider a new acceptance condition for FDFAs, which we term \emph{duo-normalization}, which considers decompositions $(uv^i,v^j)$ where in addition $v^j$ closes a loop in the state it arrives at in the respective progress DFA. We term FDFAs with this new type of acceptance \emph{duo-normalized} FDFAs. Defining FDFAs with such an acceptance condition was also suggested in the future work section of~\cite{BL23}.

We show that duo-normalized FDFAs also enjoy the good quality that computing complementation, union, and intersection and answering emptiness and equivalence can be done in non-deterministic logarithmic space. In terms of succinctness we show that they can be exponentially smaller than normalized FDFAs. 

We are also interested in the problem of finding their position in the \emph{Wagner hierarchy}. It is noted in~\cite[Sec. 5]{EhlersS22} that while for $\omega$-automata there are algorithms for computing their position in the Wagner hierarchy, there is no clear way to relate a particular FDFA to its position in the Wagner Hierarchy. 

In~\cite{Wagner75} Wagner defined a complexity measure on Muller automata: the \emph{inclusion measure}. 
Wagner showed that the inclusion measure is robust in the sense that any two Muller automata for the same language (minimal or not) have the same inclusion measure. This is thus a semantic property of the language. Since the inclusion measure is unbounded it induces an infinite hierarchy. The position on the Wagner hierarchy has been shown to be tightly related to the minimal number of colors required in a parity automaton, and the minimal number of pairs required in a Rabin/Street automaton. Deterministic B\"uchi and coB\"uchi (which are less expressive than deterministic Muller/Rabin/Streett/parity automata, that are capable of recognizing all the $\omega$-regular languages) lie in the bottom levels of the hierarchy. Given an automaton of one of the common types (B\"uchi, coB\"uch, Muller, Rabin, Streett, parity), one can compute its position in the Wagner Hierarchy in polynomial time~\cite{WilkeY96,CartonM99,PerrinPinBook}. 

We develop a syntactic notion of a measure on FDFA, that we term \emph{the diameter measure}. Loosely speaking it relates to chains of prefixes $u \preceq v_1 \prec v_2 \prec \ldots \prec v_k$ such that $u(v_i)^\omega\in L$ iff $u(v_{i+1})^\omega \notin L$, and moreover, each of the words $v_i$ is \emph{persistent} in the progress DFA of $u$. The precise definition of the term \emph{persistent} and \emph{persistent chains} is deferred to \autoref{sec:meas}. We show there that this measure is robust in the sense that computing it on two FDFAs for the same language will give the same result. The proof is by relating it to the position on the Wagner hierarchy. 
We show that computing the Wagner position of a duo-normalized FDFA can be done in PSPACE and it is PSPACE-complete, whereas for normalized FDFAs this computation can be done in non-deterministic logspace. So this is one place where the added succinctness of duo-normalized FDFAs comes at a price. 

The rest of the paper is organized as follows. We give some basic definitions and explain the Wagner hierarchy in \autoref{sec:prelim}.
We introduce duo-normalized FDFAs in \autoref{sec:fdfa-acc} where we show that it is not more expensive to compute the Boolean operations on them, or to answer emptiness and equivalence about them. \autoref{sec:meas} is devoted to defining the \emph{diameter measure} and proving that its computation is PSPACE-complete. \autoref{sec:nat-col-succ} relates to the recent works of~\cite{EhlersS22,BL23} that defines natural colors for infinite and finite words wrt to a language. This section  discusses a new canonical model for FDFAs, termed the \emph{colorful FDFA} and provides several succinctness results, relating duo-normalized FDFAs, the colorful FDFA, and previously studied canonical FDFAs. We conclude with a short discussion in \autoref{sec:conclusions}.
Due to space limitations, some proofs are deferred to the appendix.

\section{Preliminaries}\label{sec:prelim}

A (complete deterministic) \emph{automaton structure} is a tuple $\aut{A}=(\Sigma, Q, q_0, \delta)$ consisting of an alphabet $\Sigma$, a finite set $Q$ of
states, an initial state $q_0$, and a transition function ${\delta: Q \times \Sigma \rightarrow Q}$. 
A run of an automaton on a finite word ${v=a_1 a_2\ldots a_n}$ is a sequence of states ${q_0,q_1,\ldots,q_n}$ starting with the initial state such that for each $i\geq 0$, ${q_{i+1}=\delta(q_i,a_i)}$.  A run on an infinite word is defined similarly and results in an infinite sequence of states. Let $\aut{A}=( \Sigma, Q, q_0,\delta)$ be an automaton. We use $\autstate{A}{w}$ for the state that $\aut{A}$ reaches on reading $w$.

By augmenting an automaton structure with an acceptance condition $\alpha$, obtaining a tuple $( \Sigma, Q, q_0,$ $\delta, \alpha )$, we get an \emph{automaton}, a machine that accepts some words and rejects others. An automaton 
accepts a word if the run on that word is accepting. For finite words the acceptance condition is a set $F \subseteq Q$ and a run on a word $v$ is accepting if it ends in an accepting state, i.e. a state $q\in F$. For infinite words, there are various acceptance conditions in the literature.
The common ones are B\"uchi, Muller, Rabin, Streett and parity. They are all   
are defined with respect to the set of states visited infinitely often during a run. For a run $\rho=q_0q_1q_2\ldots$ we define $\inf(\rho)= \{ q \in Q ~|~ \forall i\!\in\!\mathbb{N}.\ \exists j\!>\!i.\ q_j=q\}$. 
We focus here on 
the most common  types  --- B\"uchi, Muller and parity.
  
\begin{itemize}[nosep]
\item 
	A \emph{B\"uchi} acceptance condition is a set $F \subseteq Q$. A run of a B\"uchi\ automaton is accepting if it visits $F$ infinitely often. That is, if $\inf(\rho)\cap F \neq \emptyset$.
\item 
	A \emph{parity} acceptance condition is a mapping $\kappa:Q\rightarrow \{0,1,\ldots,k\}$ of the states to a number (referred to as a \emph{color}). For a subset $Q'\subseteq Q$, we use $\kappa(Q')$ for the set $\{\kappa(q)~|~q\in Q'\}$. A run  $\rho$ of a parity automaton is accepting if the \textbf{minimal} color in $\kappa(\inf(\rho))$ is \textbf{even}.
\item 
	A \emph{Muller} acceptance condition is a set   $\alpha=\{F_1,\ldots,F_k\}$ where $F_i\subseteq Q$ for all $1\leq i \leq k$.
	A run $\rho$ of a Muller automaton is accepting iff $\inf(\rho)\in \alpha$. That is, if the set of states visited infinitely often by the run $\rho$ is exactly one of the sets $F_i$ specified in $\alpha$.  

\end{itemize}
We use DBA, DPA, and DMA as acronyms for deterministic complete B\"uchi, coB\"uchi, parity, and Muller automata, respectively.
We use $\sema{\aut{A}}$ to denote the set of words accepted by a given automaton $\aut{A}$. 
Two automata $\aut{A}$ and $\aut{B}$ are \emph{equivalent} if $\sema{\aut{A}}=\sema{\aut{B}}$. 
Let $\aut{A}=( \Sigma, Q, q_0,\delta,F)$ be a DFA. Let $q,q'\in Q$. 
We use $\auttoq{A}{q}$ for $( \Sigma, Q, q_0,\delta,\{q\})$, namely a DFA for words reaching state $q$. We use 
$\autfromq{\aut{A}}{q}$ for $( \Sigma, Q, q,\delta,F)$,
namely a DFA for words exiting state $q$. Last, we use 
$\autfromtoq{\aut{A}}{q}{q'}$ for $( \Sigma, Q, q,\delta,\{q'\})$, namely a DFA for the words exiting $q$ and reaching $q'$.

The \emph{right congruence relation} for an $\omega$-language $L$ relates two finite words $x$ and $y$ iff there is no infinite suffix $z$ differentiating them, that is $x\sim_L y$  (for $x,y\in\Sigma^*$) iff $ \forall z\in\Sigma^\omega.\ xz\in L \iff yz \in L$. We use  $\equivcls{u}{L}$ (or simply $[u]$ when $L$ is clear from the context) for the equivalence class of $u$ induced by $\sim_L$.
A right congruence $\sim$ can be naturally associated with an automaton structure $( \Sigma, Q, q_0, \delta )$ as follows: the set of states $Q$ are the equivalence classes of $\sim$. The initial state $q_0$ is the equivalence class $[\epsilon]$. The transition function $\delta$ is defined by $\delta([u],\sigma)=[u\sigma]$.  We use $\autofrc{\sim}$ to denote the automaton structure induced by $\sim$.

Following~\cite{EhlersS22} a word $u'\in\Sigma^*$ is said to be a \emph{suffix-invariant} of $u\in\Sigma^*$ (in short \emph{$u$-invariant}) with respect to $L$ if $u\sim_L uu'$. That is, no suffix distinguishes between $u$ and the word obtained by concatenating $u'$ to $u$. 

\subsection{The Wagner Hierarchy}

\begin{figure}[b]
\begin{center}
\scalebox{0.5}{
    \begin{tikzpicture}[->,>=stealth',shorten >=1pt,auto,node distance=2.2cm,semithick,initial text=]
    
    \node[state,initial] (q0)                     {\state{0}};
    \node[state]         (q1) [above right of=q0] {\state{1}};
        \node[state]         (q2) [below right of=q0] {\state{2}};
        \node[state]         (q3) [right of=q1]       {\state{3}};
        \node[state]         (q4) [right of=q2]       {\state{4}};
        \node[label]         (alpha) [below of=q2, node distance=1cm] {$\quad\alpha:\{\{q_3\}, \{q_1,q_2,q_3\}\}$ };
    \node[label]            (l) [above  of=q0] {$\aut{M}$:};        

    \path (q0) edge [bend left=10]  node {$a$} (q1);
    \path (q0) edge [bend right=10] node [below] {$b$} (q2);
    \path (q1) edge                 node {$a$} (q3);
    \path (q1) edge [bend left=15]  node {$b$} (q2);	
    \path (q2) edge                 node {$b$} (q4);
    \path (q2) edge [bend left=15]  node {$a$} (q1);
    \path (q3) edge [loop right]    node {$a$} (q3);
    \path (q3) edge                 node [right, near start] {$b$} (q2);	
    \path (q4) edge [loop right]    node {$b$} (q4);
    \path (q4) edge                 node [right, near start] {$a$} (q1);
    
    \node[state,initial] (s0) [right of=q0, node distance=7cm] {\nspstate{0}{3}};
    \node[state]         (s1) [above right of=s0] {\nspstate{1}{3}};
    \node[state]         (s2) [below right of=s0] {\nspstate{2}{3}};
    \node[state]         (s3) [right of=s1]       {\nspstate{3}{2}};
    \node[state]         (s4) [right of=s2]       {\nspstate{4}{1}};
    \node[label]            (l) [above  of=s0] {$\aut{D}$:};       

    \path (s0) edge [bend left=10]  node {$a$} (s1);
    \path (s0) edge [bend right=10] node [below] {$b$} (s2);
    \path (s1) edge                 node {$a$} (s3);
    \path (s1) edge [bend left=15]  node {$b$} (s2);	
    \path (s2) edge                 node {$b$} (s4);
    \path (s2) edge [bend left=15]  node {$a$} (s1);
    \path (s3) edge [loop right]    node {$a$} (s3);
    \path (s3) edge                 node [right, near start] {$b$} (s2);	
    \path (s4) edge [loop right]    node {$b$} (s4);
    \path (s4) edge                 node [right, near start] {$a$} (s1);	

    			\node[label] (0)    [ right of=s3, node distance=4.0cm] { };
       
                    \node[label] (1p)  [below of=0, node distance=0.5cm]  {$\DMpk{1}$};
				\node[label] (1i)   [below of=1p, node distance=1.2cm]  { };
				\node[label] (1n)    [below of=1i, node distance=1.2cm] {$\DMnk{1}$};
				\node[label] (1pn)    [ right of=1i, node distance=1.2cm] {$\DMpmk{1}$};
				
				\node[label] (2p)   [ right of=1p, node distance=2.4cm]  {$\DMpk{2}$};
				\node[label] (2i)   [below of=2p, node distance=1.2cm]  { };
				\node[label] (2n)    [below of=2i, node distance=1.2cm] {$\DMnk{2}$};
				\node[label] (2pn)    [ right of=2i, node distance=1.2cm] {$\DMpmk{2}$};
				
				\node[label] (3p)   [ right of=2p, node distance=2.4cm]  {$\DMpk{3}$};
				\node[label] (3i)   [below of=3p, node distance=1.2cm]  { };
				\node[label] (3n)    [below of=3i, node distance=1.2cm] {$\DMnk{3}$};
				\node[label] (3pn)    [ right of=3i, node distance=1.2cm] {$\DMpmk{3}$};
				
				\node[label] (4p)   [ right of=3p, node distance=2.4cm]  {$\cdots$};
				\node[label] (4i)   [below of=4p, node distance=1.2cm]  { };
				\node[label] (4n)    [below of=4i, node distance=1.2cm] {$\cdots$};
				\node[label] (4pn)    [ right of=4i, node distance=1.2cm] {$\cdots$};
	
				\path (1n) edge  (1pn); 
				\path (1p) edge  (1pn); 
				\path (2n) edge  (2pn); 
				\path (2p) edge  (2pn); 
				\path (3n) edge  (3pn); 
				\path (3p) edge  (3pn); 
				
				\path (1pn) edge  (2p); 
				\path (1pn) edge  (2n); 
				\path (2pn) edge  (3p); 
				\path (2pn) edge  (3n); 
				\path (3pn) edge  (4p); 
				\path (3pn) edge  (4n);

    \end{tikzpicture}}
\end{center}	
\caption{Left, Middle: A DMA $\aut{M}$ and a DPA $\aut{D}$ for the language $L_1=L_{\infty aa \wedge \neg \infty bb}$. 
Right: The Wagner Hierarchy. An arrow from $\class{C}$ to $\class{D}$ says that $\class{C} \subsetneq \class{D}$. Note that many other strict inclusions follow by transitivity.
}\label{fig:inf-aa-fin-bb-dma}\label{fig:inf-aa-fin-bb-dpa}\label{fig-wagner-inclsion}
\end{figure}
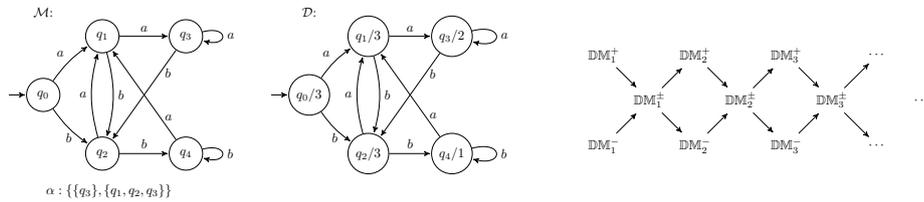

Let  
$\aut{M}=(\Sigma,Q,q_0,\delta,\alpha)$ be a complete deterministic Muller automaton. We use the term \emph{strongly connected component (SCC)} for a set of states $S\subseteq Q$ such that there is a path between every pair of states in $S$. If $S$ is a singleton $\{q\}$ we require a self-loop on $q$ for $S$ to be an SCC.
We use the term \emph{MSCC} for a \emph{maximal SCC}, that is, an SCC $S$ such that no set $S' \supset S$ is an SCC.
We say that an SCC $S\subseteq Q$ is \emph{accepting} iff $S\in\alpha$. Otherwise we say that $S$ is \emph{rejecting}. We define the \emph{positive inclusion measure} of $\aut{M}$, denoted $\picm{\aut{M}}$ to be the maximal length of an inclusion chain $S_1 \subset S_2 \subset S_3 \subset \cdots \subset S_k$ of SCCs with alternating acceptance 
where $S_1$ is accepting. (Therefore for each $1\leq i \leq k$ if $i$ is odd then $S_i$ is accepting, and if it is even then $S_i$ is rejecting.)
Likewise, we define the \emph{negative inclusion measure} of $\aut{M}$, denoted $\nicm{\aut{M}}$ to be the maximal length of an inclusion chain $S_1 \subset S_2 \subset S_3 \subset \cdots \subset S_k$ of SCCs with alternating acceptance where $S_1$ is rejecting. (Therefore for each $1\leq i \leq k$ if $i$ is odd then $S_i$ is rejecting, and if it is even then $S_i$ is accepting.)
Note that for any $\aut{M}$ the difference between $\picm{\aut{M}}$ and $\nicm{\aut{M}}$ may be at most one, since by omitting the innermost element of a chain we remain with a chain shorter by one, and of the opposite sign.

\begin{example}\label{ex:inf-aa-fin-bb-inc-ms}
\autoref{fig:inf-aa-fin-bb-dma} shows a Muller automaton for the language $L_1=L_{\infty aa \wedge \neg \infty bb}$. 
The inclusion chain 
$\{\state{1},\state{2}\} 
\subset 
\{\state{1},\state{2},\state{3}\}
\subset 
\{\state{1},\state{2},\state{3},\state{4}\}$ 
is a negative inclusion chain of size $3$ 
(since $\{\state{1},\state{2}\}$ is rejecting, $\{\state{1},\state{2},\state{3}\}$ is accepting, and $\{\state{1},\state{2},\state{3},\state{4}\}$ is rejecting). There are no negative inclusion chains of size $4$, and there are no positive inclusion chains of size $3$. (Note that $\{q_3\}\subset\{q_2,q_3\}\subset\{q_1,q_2,q_3\}$ is not an inclusion chain since $\{q_2,q_3\}$ is not an SCC.)
\end{example}

Wagner~\cite{Wagner75} showed that this measure is robust in the sense that any two DMAs that recognize the same language have the same positive and negative inclusion measures. 
\begin{theorem}[Robustness of the inclusion measures~\cite{Wagner75}]\label{thm:wagner-robust} 
Let $\aut{M}_1$, $\aut{M}_2$ be two DMAs where $\sema{\aut{M}_1}=\sema{\aut{M}_2}$.
For $i\in\{1,2\}$, let $\picm{\aut{M}_i}=p_i$ and $\nicm{\aut{M}_i}=n_i$.
Then $p_1=p_2$ and $n_1=n_2$.
\end{theorem}

Since this measure is robust and since one can construct DMAs with arbitrarily long inclusion chains, the inclusion measure yields an infinite hierarchy of $\omega$-regular languages.
Formally, the classes of the Wagner Hierarchy are defined as follows for a positive integer $k$:
\begin{center}
$\begin{array}{l@{~=~}l}
    \DMpk{k} & \{ L ~|~  \exists \mbox{ DMA } \aut{M}  \mbox{ s.t. } \sema{\aut{M}}=L \mbox{ and } \picm{\aut{M}}\leq k \mbox{ and }  \nicm{\aut{M}}< k \} 
    \\[2mm]
    \DMnk{k} & \{ L ~|~  \exists \mbox{ DMA } \aut{M}  \mbox{ s.t. } \sema{\aut{M}}=L \mbox{ and }  \picm{\aut{M}}< k \mbox{ and }  \nicm{\aut{M}}\leq k \} 
    \\[2mm]
    \DMpmk{k} & \{ L ~|~  \exists \mbox{ DMA } \aut{M}  \mbox{ s.t. } \sema{\aut{M}}=L \mbox{ and } \picm{\aut{M}}\leq k \mbox{ and } \nicm{\aut{M}}\leq k \} \\
\end{array}$
\end{center}	
\begin{example}\label{ex:WagnerHierMeasure}
Following the discussion on the inclusion measure of the DMA for $L_{\infty aa \wedge \neg \infty bb}$ given in \autoref{ex:inf-aa-fin-bb-inc-ms} we conclude that $L_{\infty aa \wedge \neg \infty bb}\in\DMnk{3}$. 
\end{example}

The hierarchy is depicted in \autoref{fig-wagner-inclsion}. Note that if $\aut{A}$ is an $\omega$-automaton, for any of the $\omega$-automata types, then it can be recognized by a Muller automaton on the same structure. Transforming a B\"uchi $\aut{B}$ automaton with accepting states $F$ to a Muller automaton $\aut{M}_{\aut{B}}$ yields an acceptance condition $\alpha_{\aut{B}}=\{F'~|~F' \cap F \ne \emptyset\}$.  Note that for any $F'\in\alpha_{\aut{B}}$ and any $F''\supseteq F'$ it holds that $F''\in\alpha_{\aut{B}}$. Therefore $\picm{\aut{M}_{\aut{B}}}=1$ and 
$\nicm{\aut{M}_{\aut{B}}}=2$ (unless $\sema{\aut{B}}=\Sigma^\omega$ or $\emptyset$). Hence all languages recognized by a DBA are in $\DMnk{2}$. 
Dually, one can see that all languages recognized by a DCA  are in $\DMpk{2}$.  It can be shown~\cite{PerrinPinBook} that
a parity automaton for a language in $\DMnk{k}$ can suffice with colors $\{1,\ldots,k\}$ if $k$ is odd, and $\{0,\ldots,k\}$ if it is even.
Likewise a DPA for for a language in $\DMpk{k}$ can suffice with colors $\{0,\ldots,k-1\}$ if $k$ is odd, and with $\{1,\ldots,k\}$ otherwise. 
A DPA in $\DMpmk{k}$  requires $k+1$ colors starting with $0$.

\begin{example}\label{ex:ex:WagnerHierMeasureInParity}
 Consider the parity automaton for 
 $L_{\infty aa \wedge \neg \infty bb}$ given in \autoref{fig:inf-aa-fin-bb-dpa}. It uses the three colors $\{1,2,3\}$ in accordance with our conclusion in \autoref{ex:WagnerHierMeasure} that
 $L_{\infty aa \wedge \neg \infty bb}\in\DMnk{3}$.
\end{example}

\section{FDFAs with duo-normalized acceptance condition}\label{sec:fdfa-acc}
As already mentioned, none of the common $\omega$-automata has a unique minimal automaton, and the number of states in the minimal automaton may be bigger than the number of equivalence classes in $\sim_L$. For example, $L_2=(\Sigma^*abc)^\omega$ has one equivalence class under $\sim_{L_2}$ since an infinite extension $xw$ for $w\in\Sigma^\omega$ of any finite word $x$ is in the language iff $w\in L_2$.

The quest for finding a correspondence between equivalence classes of the language and an automaton model lead to the development of \emph{Families of Right Congruences} (FORCs)~\cite{MalerS97} and \emph{Families of DFAs} (FDFAs)~\cite{AngluinF16}.
These definitions build on
the well-known result that two $\omega$-regular languages are equivalent if and only if they agree on the set of ultimately periodic words~\cite{Buchi62,CalbrixNP93old}, and thus consider only ultimately periodic words, i.e. words of the form $uv^\omega$. We also consider only such words, and represent them as pairs $(u,v)$ for $u\in\Sigma^*$ and $v\in\Sigma^+$.

Several canonical FDFAs were introduced over the years,
the periodic FDFA~\cite{CalbrixNP93old}, the syntactic FDFA~\cite{MalerS97}, the recurrent FDFA~\cite{AngluinF16}, and the limit FDFA~\cite{li2023novel}. 
We do not go into the details of their definition but summarize the succinctness relations among them.
It was shown in~\cite{AngluinF16} that the syntactic and recurrent FDFAs can be exponentially more succinct than the periodic FDFA, while the translations in the other direction are at most polynomial. Further,
  the recurrent FDFA is never bigger than the syntactic and can be quadratically more succinct than the syntactic FDFA~\cite{AngluinF16}. Limit FDFAs are the duals of recurrent FDFAs, and thus can also be at most quadratically bigger than the recurrent; and
  there are examples of quadratic blowups in the transformation from the recurrent to the limit or vice versa~\cite{li2023novel}.

The gain in succinctness in going from the syntactic to the recurrent (or limit) FDFAs is removing the requirement that $x \approx^u y$ implies that $ux \sim uy$ which comes from the definition of FORC. \footnote{A FORC is a pair $(\sim,\{\approx^u\})$ where $\sim$ is a right congruence, $\approx^u$ is a right congruence for every equivalence class $u$ of $\sim$, and it satisfies that $x \approx^u y$ implies that $ux \sim uy$.  An $\omega$-language $L$ is recognized by a FORC
$(\sim,\{\approx^u\})$ if it can be written as a union of sets of the form $[u]([v]_u)^\omega$ such that $uv \sim_L u$. Every FORC corresponds to an FDFA, but the other direction many not hold. This is since there is no requirement on the relation between the progress DFAs and the leading automaton in an FDFA, while there is in a FORC.}
The gain in succinctness of the syntactic/recurrent/limit FDFAs compared to the periodic FDFA is due to the use of a different type of acceptance condition.

An FDFA is a pair $\aut{F}=(\aut{Q},\{\prog{q}{}\}_{q\in Q})$ consisting of a \emph{leading}
automaton structure $\aut{Q}$ and of a \emph{progress} DFA $\prog{q}{}$
for each state $q$ of $\aut{Q}$. There are a few ways to define acceptance on FDFAs. An $\omega$-word $w$ is accepted by an FDFA using $\textsc{a}$-acceptance if there exists an $\textsc{a}$-decomposition of $w$ into $(u,v)$ such that $v$ is accepted by the progress DFA $\prog{\aut{Q}(u)}{}$ corresponding to the state $\aut{Q}(u)$ that the leading automaton reaches after reading $u$. 
We henceforth use $\prog{u}{}$ for $\prog{\aut{Q}(u)}{}$.

In exact acceptance, that is used in the periodic FDFA, any decomposition of the $\omega$-word into an ultimately periodic word is considered. In normalized acceptance, used by the other three canonical FDFAs, only decompositions $(u,v)$ in which the periodic part $v$ loops in the leading automaton, i.e. $\aut{Q}(u)=\aut{Q}(uv)$ are considered. 

As shown in~\cite{AngluinF16} this acceptance condition, termed \emph{normalization}, can yield an exponential save in the number of states. The intuition is that some periods are easier to verify as good periods if one considers some repetitions of them. For instance, in the language $(\bigcup_{1\leq i \leq n} ( i \cdot (\Sigma\setminus\{i\})^* \cdot i))^\omega$ over $\Sigma=\{\sharp\}\cup\{1,\ldots,n\}$ it is harder to figure out that $(\epsilon,12)$ should be accepted than that $(\epsilon, 121\!\cdot\! 212)$ though both represent the same $\omega$-word. 

For similar reasons one might wonder if considering decompositions that also close a loop in the progress automaton, might as well lead to an exponential saveup.  In the following we define FDFAs with such an acceptance condition, which we term \emph{duo-normalization}. Considering FDFAs with such an acceptance condition was also proposed in the future work section of~\cite{BL23}, where it is termed \emph{idempotent}.

\begin{definition}[$\omega$-words decomposition wrt an FDFA]
Let $u\in\Sigma^*$, $v\in\Sigma^+$ and $w\in\Sigma^\omega$.
Let $\aut{F}=(\aut{Q},\{\prog{q}{}\}_{q\in Q})$ be an FDFA.
\begin{itemize}[nosep]
    \item $(u,v)$ is a decomposition of $w$ if $uv^\omega=w$.
    \item A decomposition $(u,v)$ is \emph{normalized} if $\aut{Q}(uv)=\aut{Q}(u)$.
    \item A normalized decomposition is \emph{duo-normalized} if $\aut{P}_u(vv)=\aut{P}_u(v)$.  
\end{itemize}
\end{definition}

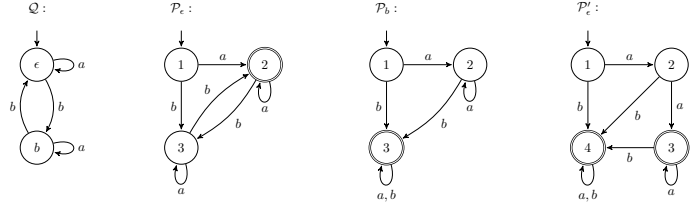
\begin{figure}[t]
\begin{center}
\scalebox{0.5}{
    \begin{tikzpicture}[->,>=stealth',shorten >=1pt,auto,node distance=2.2cm,semithick,initial text=,  initial above]

    \node[state,initial]    (Q1)        {$\epsilon$};
    \node[state]    (Q2) [below of=Q1]       {$b$};
    
    \node[label]            (labelQ) [above of=Q1, node distance=1.5cm] {$\aut{Q}:$};
    
    \path (Q1) edge [loop right]             node {$a$} (Q1);
    \path (Q2) edge [loop right]             node {$a$} (Q2);
    \path (Q1) edge [bend left]             node {$b$} (Q2);
    \path (Q2) edge [bend left]             node {$b$} (Q1);

    \node[state,initial]    (p1) [right of=Q1, node distance=3.8cm]       {$\statep{1}$};
    \node[state,accepting]    (p2) [right of=p1]       {$\statep{2}$};
    \node[state]    (p3) [below of=p1]       {$\statep{3}$};

    \node[label]            (labelPe) [above of=p1, node distance=1.5cm] {$\aut{P}_\epsilon:$};

    \path (p1) edge               node {$a$} (p2);
    \path (p1) edge               node [left] {$b$} (p3);
    \path (p2) edge [loop below]             node {$a$} (p2);
    \path (p2) edge [bend left=15]             node {$b$} (p3);
    \path (p3) edge [loop below]             node {$a$} (p3);
    \path (p3) edge [bend left=15]             node {$b$} (p2);

    \node[state,initial]    (pp1) [right of=Q1, node distance=9.2cm]       {$\statepp{1}$};
    \node[state]    (pp2) [right of=pp1]       {$\statepp{2}$};
    \node[state,accepting]    (pp3) [below of=pp1]       {$\statepp{3}$};

    \node[label]            (labelPe) [above of=pp1, node distance=1.5cm] {$\aut{P}_b:$};

    \path (pp1) edge               node {$a$} (pp2);
    \path (pp1) edge               node [left] {$b$} (pp3);
    \path (pp2) edge [loop below]             node {$a$} (pp2);
    \path (pp2) edge [bend left=15]             node {$b$} (pp3);
    \path (pp3) edge [loop below]             node {$a,b$} (pp3);

    \node[state,initial]    (ppp1) [right of=Q1, node distance=14.5cm]       {$\stateppp{1}$};
    \node[state]    (ppp2) [right of=ppp1]       {$\stateppp{2}$};
    \node[state,accepting]    (ppp3) [below of=ppp2]  {$\stateppp{3}$};
    \node[state,accepting]    (ppp4) [below of=ppp1]       {$\stateppp{4}$};

    \node[label]            (labelPe) [above of=ppp1, node distance=1.5cm] {$\aut{P}'_\epsilon:$};

    \path (ppp1) edge               node {$a$} (ppp2);
    \path (ppp1) edge               node [left] {$b$} (ppp4);
    \path (ppp2) edge              node {$a$} (ppp3);
    \path (ppp2) edge              node {$b$} (ppp4);
    \path (ppp4) edge [loop below]             node {$a,b$} (ppp4);
    \path (ppp3) edge  [loop below]            node {$a$} (ppp3);
    \path (ppp3) edge              node {$b$} (ppp4);

   \end{tikzpicture}}
\end{center}
\vspace{-5mm}
\caption{Two FDFAs $\aut{F}_1=(\aut{Q},\{\aut{P}_\epsilon,\aut{P}_b\})$ and $\aut{F}_2=(\aut{Q},\{\aut{P}'_\epsilon,\aut{P}_b\})$ for the language $(\Sigma^*b)^\omega\cup(bb)^*a^\omega$ using normalized and duo-normalized acceptances, respectively.}\label{fig:simple-fdfa}
\end{figure}

\begin{definition}[Exact, Normalized, and Duo-Normalized acceptance]
Let $\aut{F}=(\aut{Q},\{\prog{q}{}\}_{q\in Q})$ be an FDFA, $u{\in}\Sigma^*$, $v{\in}\Sigma^+$. 
We define three types of acceptance conditions:
We say that $(u,v)\in\sema{\aut{F}}$ using \emph{exact-acceptance} if
$v\in\sema{\prog{u}{}}$.
 We say that $(u,v)\in\sema{\aut{F}}$ using \emph{normalized} (resp. duo-normalized) acceptance if there exists a normalized (resp. duo-normalized) decomposition $(x,y)$ of $uv^\omega$ such that $y\in\sema{\prog{x}{}}$.
\end{definition}
An FDFA $\aut{F}$ is said to be $\textsc{a}$-\emph{saturated} if for every ultimately periodic word $w$, all its $\textsc{a}$-decompositions agree on membership in $\aut{F}$.
Assuming saturation, and an efficient $\textsc{a}$-normalization process (as suggested by \autoref{clm:duo-exists})
we can alternatively define $\textsc{a}$-acceptance as in~\cite{AngluinF16} using the efficient procedure that given any $(u,v)$ returns a particular $(x,y)$ that is $\textsc{a}$-normalized and satisfies $uv^\omega=xy^\omega$.
Henceforth, all FDFAs are presumed to be saturated.

\begin{restatable}{claim}{clmduoexists}\label{clm:duo-exists}
For every $x \in\Sigma^*$ and $y \in\Sigma^+$ the word $xy^\omega$ has an 
$\textsc{a}$-decomposition of the form $(xy^i,y^j)$ where $i$ and $j$ are of size quadratic in $\aut{F}$ for every $\textsc{a}\in\{$exact, normalized, duo-normalized$\}$.
\end{restatable}

Note that if $\aut{F}$ is saturated using exact-acceptance and it recognizes $L$, then using normalized-acceptance it is also saturated and recognizes $L$. If $\aut{F}$ is saturated using normalized-acceptance and it recognizes $L$, 
then using duo-normalized acceptance it is also saturated and recognizes $L$.

\begin{example}
\autoref{fig:simple-fdfa} shows 
two FDFAs. The FDFA $\aut{F}_1=(\aut{Q},\{\aut{P}_\epsilon,\aut{P}_b\})$ has a leading automaton with two states $[\epsilon]$ and $[b]$, and the corresponding progress automata are $\aut{P}_\epsilon$ and $\aut{P}_b$. Consider the word $a^\omega$. Since $(\epsilon,a)$ is a normalized decomposition of $a^\omega$  and $a\in\sema{\aut{P}_\epsilon}$, using normalized acceptance $a^\omega$ is accepted in $\aut{F}_1$. The FDFA $\aut{F}_2=(\aut{Q},\{\aut{P}'_\epsilon,\aut{P}_b\})$ uses duo-normalization and $\aut{P}'_\epsilon$ as the progress DFA for $[\epsilon]$.
 Since $(\epsilon,a)$ and $(\epsilon,aa)$ are resp. a normalized and duo-normalized decomposition of $a^\omega$ wrt $\aut{F}_2$, using normalized acceptance $a^\omega$ should be rejected, but using duo-normalization it should be accepted. In this example the FDFA using duo-normalization has more states. Later on we provide an example where an FDFA using duo-normalization has  fewer states, and even exponentially fewer.
\end{example}

\begin{restatable}{theorem}{clmduonormmainainsallgood}\label{clm:duo-norm-good-prop}
The following holds for saturated FDFAs using  duo-normalized acceptance.
Complementation can be computed in constant space; intersection, union and membership can be computed in logarithmic space; and emptiness, universality, containment and equivalence can be computed in non-deterministic logarithmic space.
\end{restatable}

From now on, unless otherwise stated, we work with duo-normalization. That is when we say $(u,v)\in\sema{\aut{F}}$ or $w\in\sema{\aut{F}}$ we mean according to duo-normalization.

\section{The Diameter Measure - A Robust Measure on FDFAs}\label{sec:meas}

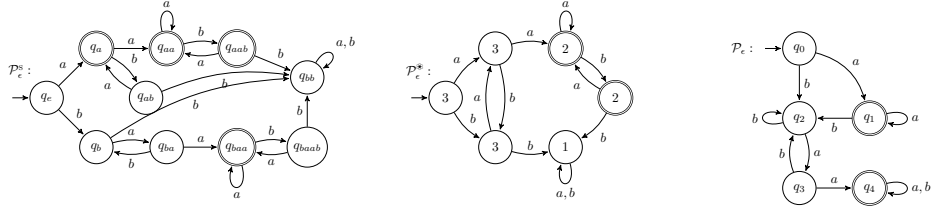
\begin{figure}
\begin{center}
\scalebox{0.5}{
    \begin{tikzpicture}[->,>=stealth',shorten >=1pt,auto,node distance=1.85cm,semithick,initial text=, initial left]

    \node[state,initial]    (e)                    {\state{e}};
    \node[state, accepting] (a) [above right of=e] {\state{a}};
    \node[state,accepting]  (aa) [right of=a]      {\state{aa}};
    \node[state]            (ab) [below right of=a] {\state{ab}};
    \node[state,accepting]  (aab) [right of=aa]     {\state{aab}};   
    \node[state]            (b) [below right of=e] {\state{b}};
    \node[state]            (ba) [right of=b]      {\state{ba}};
    \node[state,accepting]  (baa) [right of=ba]    {\state{baa}};    
    \node[state]            (baab) [right of=baa]  {\state{baab}};     
    \node[state]            (bb) [above of=baab]   {\state{bb}};
    \node[label]            (alpha) [above left of=e, node distance=1cm] {$\aut{P}^{\syntactic}_\epsilon:$};

    \path (e) edge              node {$a$} (a);
    \path (e) edge              node {$b$} (b);
    \path (a) edge    node  {$a$} (aa);
    \path (a) edge  [bend left=15]  node {$b$} (ab);	
    \path (aa) edge [loop above]   node {$a$} (aa);
    \path (aa) edge [bend left=15]      node {$b$} (aab);
    \path (aab) edge  [bend left=15]    node {$a$} (aa);
    \path (aab) edge      node {$b$} (bb);
    \path (ab) edge [bend left=15]  node {$a$} (a);
    \path (ab) edge  [bend left=15]    node [below] {$b$} (bb);
    \path (b) edge   [bend left=15]              node {$a$} (ba);
    \path (b) edge    [bend left=15]            node [below] {$b$} (bb);            
    \path (ba) edge                 node   {$a$} (baa);	
    \path (ba) edge  [bend left=15]  node {$b$} (b);
    \path (baa) edge    [loop below]             node  {$a$} (baa);
    \path (baa) edge   [bend left=15]               node  {$b$} (baab);
    \path (baab) edge  [bend left=15]                node  {$a$} (baa);
    \path (baab) edge                 node  {$b$} (bb);
    \path (bb) edge   [in=30,out=60, loop,looseness=7]    node  {$a,b$} (bb);

    \node[state,initial]    (qe)  [right of=e, node distance=10.5cm]  {3};
    \node[state] (qa)  [above right of=qe] {3};
    \node[state, accepting] (qaa)  [right of=qa] {2}; 
    \node[state,accepting]  (qaab)[below right of=qaa]     {2};
    \node[state]            (qb)  [below right of=qe] {3};
    \node[state]  (qbb) [right of=qb]       {1};   
    \node[label]         (qalpha) [above left of=qe, node distance=1cm] {$\aut{P}^{\colorful}_\epsilon:$};

    \path (qe) edge   [bend left=15]           node {$a$} (qa);
    \path (qe) edge   [bend right=15]           node {$b$} (qb);
    \path (qa) edge   [bend left=15] node  {$a$} (qaa);
    \path (qa) edge  [bend left=15]  node {$b$} (qb);	
    \path (qb) edge  [bend left=15]  node  {$a$} (qa);
    \path (qb) edge  [bend right=15]  node {$b$} (qbb);	
    \path (qaa) edge  [loop above]  node  {$a$} (qaa);
    \path (qaa) edge  [bend left=15]  node {$b$} (qaab);	
    \path (qaab) edge  [bend left=15]  node  {$a$} (qaa);
    \path (qaab) edge  [bend left=15]  node {$b$} (qbb);	
    \path (qbb) edge   [loop below]              node  {$a,b$} (qbb);

    \node[state,initial]    (sqe)    [right of=qa, node distance=8cm]    {$q_0$};
    \node[state]    (sqab) [below  of=sqe]       {$q_2$};
    \node[state,accepting]    (sqa) [ right of=sqab]       {$q_1$};
    \node[state]    (sqaba) [below of=sqab]       {$q_3$};
    \node[state,accepting]    (sqabaa) [right of=sqaba]       {$q_4$};
    \node[label]            (slabelQ) [left of=sqe, node distance=1.5cm] {$\aut{P}_\epsilon:$};
    
    \path (sqe) edge [bend left]              node {$a$} (sqa);
    \path (sqe) edge             node {$b$} (sqab);
    
    \path (sqa) edge [loop right]            node {$a$} (sqa);
    \path (sqa) edge             node {$b$} (sqab);
    
    \path (sqab) edge  [bend left=20]           node {$a$} (sqaba);
    \path (sqab) edge [loop left]            node {$b$} (sqab);
    
    \path (sqaba) edge               node {$a$} (sqabaa);
    \path (sqaba) edge  [bend left=20]             node {$b$} (sqab); 
    
    \path (sqabaa) edge [loop right]             node {$a,b$} (sqabaa);
    \end{tikzpicture}}
\end{center}	
\caption{Left, Middle: Two FDFAs $\aut{F}^{\syntactic}=(\aut{Q},\{\aut{P}^{\syntactic}_\epsilon\})$ and $\aut{F}^{\colorful}=(\aut{Q},\{\aut{P}^{\colorful}_\epsilon\})$ for the language $L_{\infty aa \wedge \neg \infty bb}$ where $\aut{Q}$ is a one-state leading automaton. $\aut{F}^\syntactic$ uses normalized acceptance, $\aut{F}^\colorful$ uses duo-normalized acceptance.
Right: The progress DFA $\aut{P}_\epsilon$ for an FDFA accepting $\infty aa$ that uses duo-normalization and a one-state leading automaton.}\label{fig:inf-aa-fin-bb-fdfas}\label{fig:duo-nor-chain-longer-than-wagner}
\end{figure}

In the following section we define a measure on FDFAs that is tightly related to the inclusion measure of the Wagner hierarchy. The defined measure is robust among FDFAs in the same way that the inclusion measure is robust among DMAs. That is, every pair of FDFAs $\aut{F}_1$ and $\aut{F}_2$ recognizing the same language agree on this measure.

To devise this measure we would like to understand what does
 the inclusion measure entail on an FDFA for the language.  If a DMA has an inclusion chain $S_1\subset S_2 \subset S_3$
there is a state $q_u$ in $S_1$ reachable by some word $u$ from which there are words $v_1,v_2,v_3$ looping on $q_u$ while traversing the states of $S_1$, $S_2$ and $S_3$, respectively (all and only these states). Assume $S_1$ is negative, then $u(v_1)^\omega\notin L$, $u(v_2)^\omega\in L$ and $u(v_3)^\omega\notin L$. 
Since they all loop back to $q_u$,  we have that $u(v_1v_2)^\omega$ also loops in $S_2$ and is thus accepted and $u(v_1v_2v_3)^\omega$ also loops in $S_3$ and is thus rejected. 
Since $v_1\prec v_1v_2 \prec v_1v_2v_3$ (where $\prec$ denotes the prefix relation) tracing the run on $v_1v_2v_3$ in a progress DFA for $u$  we expect the state reached after $v_1$ to be rejecting, the one after $v_1v_2$ to be accepting and the one after $v_1v_2v_3$ to be rejecting. 
To be precise, we should expect this only if the words $v_1$, $v_1v_2$ and $v_1v_2v_3$ are $\textsc{a}$-normalized, where $\textsc{a}$ is the normalization used by the  FDFA.

Let's inspect this on our running example $L_{\infty aa \wedge \neg \infty bb}$ with the DMA in \autoref{fig:inf-aa-fin-bb-dma}
and the inclusion chain $S_1=\{\state{1},\state{2}\}$, $S_2=
\{\state{1},\state{2},\state{3}\}$, and $S_3=
\{\state{1},\state{2},\state{3},\state{4}\}$. 
We can choose $q_1$ for the pivot state $q_u$ of $S_1$ and the words $v_1=ba$, $v_2=aba$ and $v_3=abba$, that loop respectively in $S_1,S_2,S_3$.
\autoref{fig:inf-aa-fin-bb-fdfas} (Left, Middle) provides two FDFAs for this language.
The FDFA
$\aut{F}^\syntactic$ 
uses normalization and $\aut{F}^\colorful$ uses duo-normalization. 
Looking at $\aut{P}^\syntactic_\epsilon$,
we see that $v_1$, $v_1v_2$ and $v_1v_2v_3$ are normalized and the acceptance of the states we land in after reading them ($\state{ba}$, $\state{baa}$, $\state{bb}$) are rejected/accepted as we expect. The same is true for $\aut{P}^\colorful_\epsilon$. 

Should we entail from this discussion and example that the maximal number of alternations between rejecting and accepting states along any path in an FDFA for a language in $\DMnk{k}$ to be at most $k-1$? This is true in the progress DFA $\aut{P}^{\colorful}_{\epsilon}$, but
the progress DFA $\aut{P}^\syntactic_\epsilon$ clearly refutes it, since it has strongly connected accepting and rejecting states (e.g. $\state{a}$ and $\state{ab}$) and so we can create paths with an unbounded number of alternations between them. 

Take such a path with say $k+1$ prefixes
$z_1\prec z_2 \prec \ldots \prec z_{k+1}$
alternating between accepting and rejecting states. Are the words $z_i$'s normalized?
They can be. Take for instance $z_1=a$, $z_2=ab$, and so on ($z_{k+1}=z_k \cdot b$ if $k$ is even and $z_{k+1}=z_k \cdot a$ otherwise).

Can they all be duo-normalized? 
They can be as we show in \autoref{fig:duo-nor-chain-longer-than-wagner} (Right).
It shows a progress DFA $\aut{P}_\epsilon$ for an FDFA using duo-normalization and a one-state leading automaton. The FDFA recognizes the language $\infty aa$. The words $a \prec ab\prec abaa$ are all duo-normalized, and pass through alternating accepting/rejecting states though the language is in $\DMnk{2}$ and uses only two colors.
To fix this issue we introduce the notion of a \emph{persistent decomposition}.

\begin{definition}[persistent decomposition wrt an FDFA]
Let $u\in\Sigma^*$, $v\in\Sigma^+$.
Let $\aut{F}=(\aut{Q},\{\prog{q}{}\}_{q\in Q})$ be an FDFA.
A duo-normalized decomposition $(u,v)$ is \emph{persistent} if for every $z\in\Sigma^*$ there exists an $i$ such that  $\aut{Q}_u(zv)= \aut{Q}_u(zv^i)$.
\footnote{As in \autoref{clm:duo-exists} there exists $i$ and $j$ of size quadratic in $\aut{F}$ such that $(uv^i,v^j)$ is persistent. See \autoref{clm:perst-exists}.}
\end{definition}

One might try to define an additional acceptance condition using the persistent-normalization in the same manner.
But as stated by the following claim this is unnecessary.
\begin{restatable}{claim}{clmpersaccunnecessary}\label{clm:pers-acc-un-necessary}
Every FDFA defined with persistent-normalized acceptance recognizes the same language when defined with duo-normalized acceptance instead. 
\end{restatable}

Back to the issue of finding the position on the Wanger Hierarchy, we can look for chains of prefixes with alternating acceptance in the language such that each prefix in the chain is persistent.
\begin{definition}[Persistent Chain]
Let $\aut{F}$ be an FDFA and $u\in\Sigma^*$. 
We say that $v_1\prec v_2 \prec \ldots \prec v_k$ for $v_i\in\Sigma^*$
is a \emph{persistent chain of length $k$} wrt $u$ in $\aut{F}$
if $(u,v_i)$ is persistent for every $1\leq i \leq k$  and $(u,v_{i+1})\in \aut{F}$  iff  $(u,v_i)\notin \aut{F}$ for every $1\leq i < k$. We say the chain is \emph{positive} (resp. \emph{negative}) if $(u,v_1)\in\aut{F}$ (resp. $(u,v_1)\notin\aut{F}$).
\end{definition}

We are now ready to state the measure on FDFA that relates them to the Wagner Hierarchy.
\begin{definition}[The Diameter Measure]
We define the \emph{positive} (resp. \emph{negative}) \emph{diameter measure} of a progress DFA $\aut{D}$, denoted $\pdm{\aut{D}}$ (resp. $\ndm{\aut{D}}$) as the maximal $k$ such that there exists a \emph{positive} (resp. negative) persistent chain of length $k$ in $\aut{D}$. 
We define $\pdm{\aut{F}}$ as $\max \{ \pdm{\aut{P}_q}\colon {q\in Q}\}$ and   $\ndm{\aut{F}}$ as $\max \{ \ndm{\aut{P}_q} \colon {q\in Q}\}$.
\end{definition}

We show that the diameter measure is robust among any FDFAs for the language by relating it first to the Wagner hierarchy as formally stated below.

\begin{restatable}[Correlation to Wagner's Hierarchy]{theorem}{thmfdfawagnercor}\label{thm:fdfa-wagner-cor}
Let $\aut{F}$ be an FDFA using any of the acceptance types $\textsc{a}\in\{$exact, normalized, duo-normalized$\}$.
\begin{itemize}
    \item   $\sema{\aut{F}}\in\DMpmk{k}~$ iff $~\pdm{\aut{F}}\leq k \mbox{ and } \ndm{\aut{F}}\leq k $
    
    \item $\sema{\aut{F}}\in\DMpk{k}~$ iff $~\pdm{\aut{F}}\leq k \mbox{ and }  \ndm{\aut{F}}< k$
    
    \item $\sema{\aut{F}}\in\DMnk{k}~$ iff $~\ndm{\aut{F}}< k \mbox{ and }  \ndm{\aut{F}}\leq k$
\end{itemize} 
\end{restatable}

The proof of \autoref{thm:fdfa-wagner-cor} relies on the following lemma.
\begin{restatable}
[Pumping Persistent Periods]{lemma}{clmperspump}\label{clm:pers-pump}
    Let $\aut{A}$ be an automaton and let $v \in \Sigma^+$ be $\aut{A}$-persistent. \footnote{We say that $v$ is $\aut{A}$-persistent 
    if $\aut{A}(v^2)=\aut{A}(v)$ and for every $z \in \Sigma^*$ there exists an $i$ such that  $\aut{A}(zv) = \aut{A}(zv^i)$.
    } 
    There exists $k \in \mathbb{N}$ such that for every extension $z \in \Sigma^*$, if $vz$ is $\aut{A}$-persistent ($\aut{A}$-duo-normalized) then  
    $\autstate{A}{vz}=\autstate{A}{v^{1+k}z}$ and $v^{1+k}z$ is also persistent (duo-normalized). The same is true for $v^{1+ik}z$ for all $i\geq 1$.
\end{restatable}
Recall that a word $v$ is persistent if it closes a loop on the state it reaches (from the initial state of the respective progress automaton $\prog{u}{}$) and from every other state, upon reading $v$ it reaches its final loop when reading $v$ repetitively.
The proof of the lemma essentially shows that we can strengthen this by finding an $i$ such that, from every state, $v^i$ reaches a state on which it immediately closes a loop.

\begin{proof}[Proof of \autoref{thm:fdfa-wagner-cor}]
We prove the claim regarding the positive measure. The claim regarding the negative measure is proven symmetrically.
We show that \begin{inparaenum}
\item 
$\picm{L}\geq k$ implies $\pdm{\aut{F}}\geq k$ and
\item
$\pdm{\aut{F}}\geq k$ implies $\picm{L}\geq k$.
\end{inparaenum}
The two claims together entail that $\pdm{\aut{F}}=\picm{L}$.

\begin{inparaenum}
\item 
We start by showing that if $\picm{L}\geq k$ then $\pdm{\aut{F}}\geq k$.
From $\picm{L}\geq k$ we know that there exists an MSCC of $\aut{M}$, a DMA for $L$, subsuming SCCs  $S_1,S_2,\ldots, S_k$
such that $S_1 \subsetneq S_2 \subsetneq \ldots \subsetneq S_k$ 
and $S_i$ is an accepting component if and only if $i$ is odd.
Pick a state $s$ in $S_1$. For $1\leq i \leq k$ let $v_i$
be a word that loops on $s$ while traversing all the states of $S_i$ and no other states. Let $u$ be a word reaching $s$ from the initial state. Consider the progress DFA of $u$, $\prog{u}{}$. 
By \autoref{clm:perst-exists} there exists $l_1$ such that $y_1 = (v_1)^{l_1}$ is $\prog{u}{}$-persistent. Similarly, there exists $l_2$ such that $y_2 = (y_1v_2)^{l_2}$ is $\prog{u}{}$-persistent.
In the same way we define $y_i=(y_{i-1}v_i)^{l_i}$ for every $i\leq k$.
Consider the words $w_i=u(y_i)^\omega$ for $1\leq i \leq k$. Since the set of states visited infinitely often when reading $w_i$ is exactly $S_i$ it follows that $w_i$ is in $L$ if and only if $i$ is odd.
Since all the infixes $y_i$ loop back to $s$ it follows that all the $y_i$'s are $u$-invariants and thus persistent in $\aut{F}$. 
Since $y_1\prec y_2 \prec \ldots \prec y_k$
we have found a positive alternating persistent chain of length $k$.
Hence, $\pdm{\prog{u}{}}\geq k$ which in turn gives that $\pdm{\aut{F}}\geq k$.

\item Next we show that 
if $\pdm{\aut{F}}\geq k$ then $\picm{L}\geq k$.
Assume for some $u$ we have $\pdm{\prog{u}{}}\geq k$. Then there exists 
a persistent chain $v_1 \prec v_2 \prec \ldots \prec v_k$ of length $k$ in $\prog{u}{}$, starting with an accepting state.
Let $q_1,q_{2},\ldots,q_k$ be alternating states in such a path.  
Then $q_i$ is accepting iff $i$ is odd.

Let $\aut{M}$ be a DMA for $L$ and assume $\aut{M}$ has $n$ states. Let $k_i$ be the number promised by \autoref{clm:pers-pump} for $v_i$, and let  $l_i>n$ be such that $l_i=1+j_ik_i$ for some $j_i\in\mathbb{N}$. 
Consider $(v_1)^{l_1}$. Since $v_1$ is persistent it follows that $v_1$ loops on $q_1$ and hence so does $(v_1)^{l_1}$. Since the run on $v_2$ passes through $q_1$ after reading the prefix $v_1$, it follows that $((v_1)^{l_1} v_{2})$ reaches $q_2$, the same state that $v_2$ does, and is also persistent. 
Continuing in the same manner, let $y_i=
(((v_1)^{l_1} v_{2})^{l_2} \ldots v_i)^{l_i}$
then $y_i$ us persistent.
If follows that $w_i=u(v_i)^\omega$ is in $L$ iff $i$ is odd.
For every such $i$, consider the run of $\aut{M}$ on $w_i$, and let $\inf(w_i)=S_i$ be the states of the SCC that $\aut{M}$ eventually traverses in.
Since (1) $y_1 \prec y_2 \prec \ldots \prec y_k$ and (2) in $y_i$ the word $y_{i-1}$ is repeated at least $n$ times and (3)  
$n$ bounds the number of states of $\aut{M}$ it follows that $S_1 \subseteq S_2 \subseteq \ldots \subseteq S_k$.
From the acceptance of the words $w_i$ we conclude $S_i$ is accepting iff $i$ is odd.
Therefore the inclusions are strict, namely $S_1 \subsetneq S_2 \subsetneq \ldots \subsetneq S_k$.
This proves $\picm{L}\geq k$.
\end{inparaenum}
\end{proof}

With this theorem in place we conclude  the robustness of the diameter measure.

\begin{corollary}
    Let $\aut{F}_1$ and $\aut{F}_2$ be two FDFAs recognizing the same language. Then $\pdm{\aut{F}_1}=\pdm{\aut{F}_2}$ and $\ndm{\aut{F}_1}=\ndm{\aut{F}_2}$.
\end{corollary}

Next we show that given an FDFA we can compute its diameter measure in polynomial space. The proof shows that with each progress DFA $\prog{u}{}$ we can associate a DFA $\perstaut{u}$, which we term the \emph{persistent DFA}, which is roughly speaking the product of the leading automaton and a copy of the progress automaton from each of its states. The states 
of $\perstaut{u}$ can be classified into \emph{significant} and \emph{insignificant} 
where significant states are those that are reached by persistent words and only such words. The classification can easily be done by inspecting the `state vector'. The persistent DFA need not be built, instead a persistent chain can be non-deterministically guessed and verified in polynomial space (see full proof in the appendix).

\begin{restatable}{theorem}{thmcompdiametermeas}\label{thm:comp-diameter-meas}
The diameter measure of an FDFA can be computed in PSPACE.
\end{restatable}

We provide a matching lower bound.

\begin{restatable}{theorem}{thmlbcompdiametermeas}\label{thm:lb-diameter-meas}
The problem of determining whether the diameter measure of a duo-normalized FDFA is at least $k$ is PSPACE-hard.
\end{restatable}
\begin{proof}
The proof is by reduction from non-emptiness of intersection of DFAs, which is known to be PSPACE-complete.
Let $D_1,D_2,\ldots,D_k$ be $k$ DFAs over $\Sigma$.
We assume wlog that $D_i$'s do not accept the empty string.
Let $D_i=(\Sigma,Q_i,\iota_i,\delta_i,F_i)$.
We construct an FDFA with a one-state leading automaton and a progress DFA $\aut{P}$ over $\Sigma'=\Sigma \cup\{1,\ldots,k,k\!+\!1\} \cup \{\sharp\}$. The states of $\aut{P}$ are $\{s, s',s'', d\} \cup \{s_1,\ldots,s_k, s_{k+1}\} \cup \bigcup_{1\leq i\leq k} Q_i$.
The initial state of $P$ is $s$.
From $s$ there is a self-loop on $\Sigma'\setminus\{\sharp\}$ and a $\sharp$-transition to $s'$.
From $s'$ there is a $\Sigma$-transition to $s''$.
From $s''$ there is a self-loop on $\Sigma$ and a $1$-transition to $s_1$.
For every $1\leq i \leq k$, state $s_i$ has a self-loop on $\Sigma'\setminus \{\sharp\}$, an $i$-transition to $\iota_i$, and an $(i\!+\!1)$-transition to $s_{i+1}$. The states $q_i\in Q_i$ have all the transitions in $\delta_i$.
States in $F_i$ have in addition a self-loop on every $j<i$ and an $i$-transition back to $s_i$.
The states $s_{k+1}$ and $d$ have a self-loop on all letters.
All missing transitions lead to state $d$.
States of $Q_i$ are accepting iff $i$ is even.
The states $s_i$ are also accepting iff $i$ is even.
States $s,s',s'',d$ are non-accepting.

We claim that if there is a word $v$ in the intersection of all $D_i$'s,
then $\sharp v1 \prec \sharp v12 \prec \sharp v123 \prec \cdots \prec \sharp v12\cdots k \!\cdot\! k\!+\!1$ is a persistent chain.
Let $u_i=\sharp v12\ldots i $ for $1\leq i\leq k$. Then $u_i$ reaches $s_i$, and reading $u_i$ from $s_i$ we reach $s_i$ again. Thus $u_i$ is duo-normalized and $u_i$ is accepted iff $i$ is even. To see that $u_i$ is persistent note that reading $u_i$ from any state $s_j$ (whether $j=i$ or not) we reach $s_j$ since $v$ is accepted by $D_j$. Reading $u_i$ from $s$ we reach $s_1$ from which on reading $u_i$ we reach $s_1$ again.
Reading $u_i$ from any other state we reach $s_{k+1}$ and stay there forever. Thus $u_i$ is persistent for any $i$. Hence $u_1 \prec u_2 \prec u_{k+1}$ is a persistent chain of length $k+1$.

For the other direction, we claim that if there exists a persistent chain of length $k+1$ in $\aut{P}$ then the intersection of all $D_i$'s is non-empty. First, note that a period $w$ passes through $s_j$ iff $w$ contains $j$. From the structure of $\aut{P}$ it easily follows that if $w_1 \prec w_2 \prec \cdots \prec w_{k+1}$ is 
a persistent chain then $w_{k+1}$ passes through $s_{k+1}$.
Second, note that $w$ is accepted by the FDFA iff $w$ is of the form $x\sharp y\cdot 1 v_1\cdot 2 v_2 \cdot 3 v_3 \cdots i v_i$ for some even $i$ such that $x$ can be anything,  $y\in\Sigma^+$, 
$v_i$ is in $\textsl{Pref}((\Sigma'\setminus \{\sharp\})^*\sharp\, \Sigma^* \{1,2,\ldots, i\!-\!1\}^* i)$ and moreover, if $v_i = x_1 \sharp y_1 n_1 \cdot x_2 \sharp y_2 n_2 \cdots x_m \sharp y_m n_m \cdot x_{m+1} \sharp y_{m+1}$ 
where $x_i\in(\Sigma'\setminus \{\sharp\})^*$, $y_i\in\Sigma^+$ and $n_i\in\{1,2,\ldots,i\!-\!1\}^* i$ then $y$ and all the $y_i$'s for $1\leq i\leq m$ are in the language of $D_i$. Thus the last element $w_{k+1}$ of any persistent chain of length $k+1$ must start with $x\sharp y$ such that $y\in\Sigma^+$ and $y$ is accepted by all $D_i$'s.

Since we consider only saturated FDFAs we have to show that the resulting FDFA is saturated. That is, it accepts any duo-normalized representation $(u,v)$ of a word $uv^\omega$ in the language. To see that it is saturated note that the above description of the accepted periods does not depend on a particular decomposition to transient and periodic parts.
Thus, if $w=u(vz)^\omega=uv(zv)^i(zv)^j$, then $vz$ satisfies the above iff $(zv)^j$ satisfies the above.
\end{proof}

\begin{corollary}
    The problem of determining the position of an FDFA using duo-normalization in the Wagner hierarchy is PSPACE-complete.
\end{corollary}

The following proposition shows that there exists FDFAs where the smallest elements of 
a persistent chain are inevitably of exponential size.
The idea is to take the construction of the FDFA in the proof of \autoref{thm:lb-diameter-meas} and let the DFA $D_i$ count modulo the $i$-th prime.

\begin{restatable}{proposition}{propchainelmentexp}\label{prop:chain-elment-exp}
    There exists a family of languages $\{L_{n}\}$ with an FDFA with number of states polynomial in $n$ where
    any persistent chain of maximal length must include a word of length exponential in $n$.
\end{restatable}

\section{Relation to Natural Colors and Succinctness Results}\label{sec:nat-col-succ}

\subsection{Natural Colors}
Ehlers and Schewe~\cite{EhlersS22} show that given an $\omega$-regular language $L$ one can associate with every word $w\in\Sigma^\omega$ a natural color. If $w$ is given color $k$ wrt $L$, then the minimal color visited infinitely often by a so-called streamlined parity automaton~\cite{EhlersS22} for $L$ would be $k$, and there is no parity automaton for $L$ that would visit a lower color infinitely often when reading $w$.
Consider again the language $L_{\infty aa \wedge \neg \infty bb}$ requiring infinitely many $aa$ and finitely many $bb$ for which a DPA is given in \autoref{fig:inf-aa-fin-bb-dpa}.  The colors of $(ab)^\omega$, $(a)^\omega$, $(aab)^\omega$, $(aabb)^\omega$ and $(b)^\omega$ are $3$, $2$, $2$, $1$ and $1$, resp. The intuition is that the color is $1$ if $bb$ occurs infinitely often, it is $2$ if $aa$ occurs infinitely often but $bb$ occurs only finitely often, and it is $3$ if neither $aa$ nor $bb$ occur infinitely often.
For the language $L_{\infty abc}$ the color would be $1$ if $abc$ does not occur infinitely often, and $0$ otherwise.
We use $\infcolorof{w}$ to denote the color of $w$ wrt $L$ (which is not explicitly mentioned in the notation to avoid cumbersomeness).

Bohn and L\"oding provide a related definition~\cite[Def 2.]{BL23} that can be viewed as giving colors for finite words $v$ wrt to an $\omega$-regular language $L$ and an equivalence class $[u]$. 
We present it below using slightly different formulations, closer to that of~\cite{EhlersS22}.
We use $\fincolorof{u}{v}$  to denote the color of the finite word $v\in\Sigma^+$ wrt a finite word $u\in\Sigma^*$.
The definition satisfies that $\fincolorof{u}{v}$ returns $\max \{ \infcolorof{u(vz)^\omega} ~|~ z\in\Sigma^*\}$. 
Note that this gives that there might be $v',v''\in\Sigma^+$ such that $u(v')^\omega=u(v'')^\omega$ for some $u\in\Sigma^*$
though $\fincolorof{u}{v'}\neq \fincolorof{u}{v''}$. Indeed in the example of $L_{\infty aa \wedge \neg \infty bb}$ we have $(a)^\omega=(aa)^\omega$, yet $\fincolorof{\epsilon}{a}=3$ while $\fincolorof{\epsilon}{aa}=2$. The reason is that if the period contains $a$ followed by some $z$ the resulting color may be $3$ or $2$ or $1$ while if the period contains $aa$ followed by some $z$ the color can be $2$ or $1$ but it cannot be $3$ since $aa$ surely occurs infinitely often. 

In the formal definition we make use of the notion of relevant and irrelevant periods.
Let
 $u\in\Sigma^*$, and $v\in\Sigma^+$.
 If there exists a suffix $z\in\Sigma^*$ such that $vz$ is $u$-invariant then $v$ is \emph{relevant} to $u$. Otherwise it is \emph{irrelevant}.

We turn to define, given a transient part $u$ and a word $v$, corresponding to a prefix of some period, the color of $v$ wrt to $u$, which we denote
$\fincolorof{u}{v}$.

\begin{definition}[Natural color of finite and infinite words {\cite[Def 2.]{BL23},\cite[Def 1.]{EhlersS22}}]
\label{def:fin-clr}\label{def:inf-clr}
Let
$u\in\Sigma^*$, and $v\in\Sigma^+$. We define the \emph{natural color} of $v$ wrt $u$ in the following way:
\begin{itemize}[nosep]
    \item If $v$ is irrelevant to $u$ then it is $\neginfty$.
    \item Otherwise it is the \emph{minimal} color $c \geq 0$ such that \emph{for all} $z\in\Sigma^*$ for which $vz$ is $u$-invariant it holds that:
    \begin{itemize}[nosep]
        \item Either $u(vz)^\omega\in L$ iff $c$ is even.
        \item Or \emph{exists} $i > 0$ such that $\fincolorof{u}{(vz)^i} < c$.
    \end{itemize}
\end{itemize}
Let $w\in\Sigma^\omega$.
The color of $w$, denoted 
 $\infcolorof{w}$,
 is $\min \{ \fincolorof{u}{v}~|~w=uv^\omega \mbox{ and } v \mbox{ is } u\mbox{-invariant }  \}$. 
\end{definition}

Applying the definition of natural colors for finite words to the language $L_{\infty aa \wedge \neg \infty bb}$ we get  
$\fincolorof{\epsilon}{a}=3$
$\fincolorof{\epsilon}{aa}=2$ and
$\fincolorof{\epsilon}{aabb}=\fincolorof{\epsilon}{bb}=1$ as required.

\begin{remark}
Note that ignoring irrelevant words, the minimal color can be either $1$ or $0$ (e.g.  
as discussed above for $L_{\infty aa \wedge \neg \infty bb}$ it is $1$ and for $L_{\infty abc}$ it is $0$).
\autoref{def:fin-clr} colors irrelevant periods by $\neginfty$. 
It is convenient for the definition, but if we want to keep the number of colors minimal it is not necessary, and we can color such words with the minimal color.
We thus define $\finColorof{u}{v}$ that gives irrelevant words the color $\min\{ c\geq 0 ~\colon~\fincolorof{u}{v}=c,\ v\in\Sigma^+\}$  if this set is non-empty and 
$\min\{ c\geq 0 ~\colon~\fincolorof{u}{v}=c,\ u\in\Sigma^*,v\in\Sigma^+\}$ otherwise.~\footnote{This diverges slightly from the definition in \cite{BL23} where irrelevant words are colored $0$, also if $1$ is the minimal color for a relevant word.}
\end{remark}

\subsection{Reliable Words and Chains}

Using the definition of finite colors we can provide a semantic variation of the diameter measure we have defined in \autoref{sec:meas}.
The inclusion measure of Wagner discusses chains of infinite words that traverse maximal SCCs in an inclusion chain.
The respective semantic measure we seek needs to transform those infinite words to finite periods.
As discussed in \autoref{sec:meas} observing an arbitrary sequence of prefixes $v_1 \prec v_2 \prec \ldots \prec v_k$ such that 
$u(v_i)^\omega\in L$ iff $u(v_{i+1})^\omega\notin L$ does not necessarily relates to a sequence of $k$ alternations in an inclusion chain. We want to capture periods and chains of periods in which this does hold.

We say that period $v$ is \emph{stable} wrt $u$ if \emph{for all} $i > 0$ it holds that $\finColorof{u}{v} = \finColorof{u}{v^i}$. Otherwise it is \emph{unstable}. 
If $v$ is $u$-invariant and stable wrt $u$ we say that $v$ is $u$-\emph{reliable}. 
For instance, in
$L_{\infty aa \wedge \neg \infty bb}$ since $\finColorof{\epsilon}{a}=3$, $\finColorof{\epsilon}{aa}=2$ and $\finColorof{\epsilon}{a^i}=2$ for every $i\geq 2$ we get that the word $a$ is not $\epsilon$-reliable whereas $aa$ is.
A nice property of reliable periods is that they determine the color of the respective infinite word, that is, if $\finColorof{u}{v}=k$ and $v$ is $u$-reliable then $\infcolorof{u(v)^\omega}=k$.\footnote{This can be derived from claims in~\cite{BL23}. Since we use different formulations, we give direct proofs in the appendix, see  \autoref{prop:finite-colors-properties}.}
Indeed, in our example $\infcolorof{a^\omega}=2=\finColorof{\epsilon}{aa}\neq \finColorof{\epsilon}{a}=3$.

\begin{definition}[Reliable chains]
We say that $v_1 \prec v_2 \prec \ldots \prec v_k$ is a reliable chain wrt $L$ and $u$ if all $v_i$'s are reliable and $u(v_i)^\omega\in L$ iff $u(v_i)^\omega\notin L$ for $1\leq i <k$.
\end{definition}

The question is now how does the semantic definition of reliable words and reliable chains relate to the syntactic definitions
of duo-normalized/persistent and their chains. Looking at \autoref{fig:duo-nor-chain-longer-than-wagner} we can see that although $a$ is duo-normalized it is not reliable. The word $aa$ which is persistent is reliable. 
This is not a coincidence: if a word is persistent it is also reliable. This holds for FDFAs with any of the acceptance types. In normalized FDFAs a duo-normalized word is already reliable.

\begin{restatable}[Implying reliability]{proposition}{clmpersimplrel}\label{clm:pers-impl-rel}
Let $\aut{F}$ be an FDFA. 
If  $(u,v)$ is persistent wrt to $\aut{F}$ then $v$ is $u$-reliable, regardless of the acceptance type \textsc{a} of $\aut{F}$.
If $\aut{F}$ uses normalization and  $(u,v)$ is duo-normalized wrt to $\aut{F}$ then $v$ is $u$-reliable. 
\end{restatable}

The other direction is not necessarily true. For instance, $aa$ is reliable for $L_{\infty aa}$ but we can build a progress automaton different than the one in \autoref{fig:duo-nor-chain-longer-than-wagner} where $a$, $aa$, $aaa$ reach different states (e.g. by inserting two states $q_a$ and $q_{aa}$ between $q_0$ and $q_1$ in \autoref{fig:duo-nor-chain-longer-than-wagner}), and so $aa$ won't be persistent. 
However, the existence of a reliable chain in an FDFA $\aut{F}$, implies the existence of a persistent chain in $\aut{F}$, whichever acceptance condition it uses (\autoref{clm:reli-chn-impl-pres-chn}).

\begin{restatable}[A reliable chain implies a persistent chain]{proposition}{clmrelichnimplpreschn}\label{clm:reli-chn-impl-pres-chn} Let $\aut{F} = (\aut{Q},\{\aut{P}\})$ be an FDFA. 
    Let $u \in \Sigma^*$.
    If ${v_1 \prec v_2 \prec \ldots \prec v_d}$ is a reliable chain wrt $u$ then 
    there exists ${y_1 \prec y_2 \prec \ldots \prec y_d}$
    a persistent chain in $\prog{u}{}$. Moreover, $v_1\prec y_1$ and  ${u}(y_i)^\omega\in \sema{\aut{F}}$ iff ${u}(v_i)^\omega\in\sema{\aut{F}}$ for every $1\leq i \leq d$.
\end{restatable}

Hence, following \autoref{thm:fdfa-wagner-cor} the Wagner position of an FDFA can also be determined by the length of a maximal reliable chain. Following \autoref{thm:lb-diameter-meas} this is PSPACE-hard for duo-normalized FDFAs.
Note that using the same reasoning as in \autoref{prop:chain-elment-exp} we can deduce a bound on the size of a minimal reliable period for a certain color.

\begin{corollary}\label{cor:rel-chain-elment-exp}
     There exists a family of languages $\{L_{n}\}$ with a duo-normalized FDFA with number of states polynomial in $n$ where
     there are periods of a color $c$ and 
    the minimal length of a reliable period of color $c$ is exponential in $n$.
\end{corollary}

It follows from \autoref{clm:pers-impl-rel} that in normalized FDFAs a duo-normalized chain is a reliable chain.
Using this fact we show in \autoref{thm:wagner-on-normalized-fdfas} that for normalized FDFAs the position in the Wagner hierarchy can be computed in  NLOGSPACE. Therefore, in normalized FDFAs the size of a minimal reliable word witnessing a certain color cannot be exponential in the size of the FDFA.
This makes sense, since as we show in~\autoref{thm-duo-succ}, duo-normalized FDFAs can be exponentially more succinct than normalized FDFAs.

\paragraph*{The Colorful FDFA}
In \autoref{sec:succinctness} we provide some succinctness results regarding FDFAs with duo-normalized acceptance. 
In particular we consider a canonical FDFA, that we term the \emph{Colorful FDFA}, whose underlying right congruences ($\progequiv{u}{\colorful}$) differentiates between words $x$ and $y$ if there is an extension $z$ so that the respective extensions $xz$ and $yz$ disagree on the color (wrt $u$). This definition tightly relates to the \emph{precise family of weak priority mapping}~\cite[Def. 17]{BL23}.\footnote{
From the definition of $\progequiv{u}{\colorful}$ it is easy to see that all words in the same equivalence class have the same color. Thus, from $\progequiv{u}{\colorful}$ one can define a coloring function for every equivalence class of $\sim_L$. This coloring induces exactly the precise family of weak priority mapping wrt to $\sim_L$ from~\cite[Def. 17]{BL23}.}

\begin{definition}[The Colorful Equivalence Classes]\label{def:clr-equiv}
Let 
 $u,x,y\in\Sigma^*$.
We define $x\progequiv{u}{\colorful} y$ if for every $z\in\Sigma^*$  we have $\finColorof{u}{xz}=\finColorof{u}{yz}$.
\footnote{Note that colors for finite words are defined only for non-empty periods. We hence designate an equivalence class for $\epsilon$.} 
\end{definition}

The Colorful FDFA uses the automaton for $\sim_L$ for the leading automaton as the other canonical FDFAs do. For the progress DFA for $u$ it takes a DFA whose automaton structure is derived by the equivalence relation $\progequiv{u}{\colorful}$, where the accepting states are those with an even color. 
The acceptance type of the colorful FDFA is duo-normalization. 

\begin{definition}[The Colorful FDFA]
The \emph{colorful FDFA} for a language $L$, denoted $\fdfa{L}{\colorful}$ uses duo-normalized acceptance and consists of 
$(\autofrc{\sim_L},\{\prog{u}{\colorful}\})$ where
${\prog{u}{\colorful}}$ is a DFA with the automaton structure 
$\autofrc{\progequiv{u}{\colorful}}$ where state $q_v$ is accepting if $\finColorof{u}{v}$ is even.
\end{definition}

A generalization of the syntactic FDFA for a language $L$ and a right congruence $\sim$ refining $\sim_L$ is studied in \cite{BL23}, we term these \emph{$\sim$-syntactic FDFAs}.\footnote{\emph{Canonical FORC} is the term used in~\cite{BL23} for the underlying FORC.}
An FDFA with leading automaton $\aut{A}[\sim]$ is a \emph{$\sim$-syntactic FDFA} if the progress right congruences satisfy that $x\approx_u y$ iff (a) $ux\sim uy$ and (b) for every $z\in\Sigma^*$ it holds that $uxz \sim uyz$ implies $u(xz)^\omega \in L$ iff $u(yz)^\omega \in L$. It is shown in~\cite[Lemma 21]{BL23} that if $x$ is duo-normalized wrt $\sim$ and $\approx_u$ then for every $y \approx_u x$ we have that $y$ is also duo-normalized. We can thus refer to a state as being duo-normalized. It is then showed that two duo-normalized states in the same SCC of a $\sim$-syntactic progress DFA agree on acceptance~\cite[Lemma 23]{BL23}. These properties allows defining a polynomial procedure that associates the correct color with every state of a \emph{$\sim$-syntactic FDFA}~\cite{BL23}. 

It follows that on \emph{$\sim$-syntactic FDFAs} the Wagner position can be determined in polynomial time. The proof can be generalized to any FDFA using normalization in which the right congruence $x \approx_u y$ implies $x \sim_L y$. We call such FDFAs \emph{projective FDFAs}. \autoref{fig:class-incl} depicts the inclusions among these classes.
Since any FDFA using normalization can be transformed with a quadratic blowup into a \emph{projective FDFAs} (by multiplying the progress DFAs by the leading automaton) we have that the Wagner position on normalized FDFAs can be computed in polynomial time.

To relate this to our measure we note that considering FDFAs with normalized rather than duo-normalized acceptance, by \autoref{clm:pers-impl-rel} and \autoref{clm:reli-chn-impl-pres-chn}, it suffices to look for duo-normalized chains.
It is not hard to see that~\cite[Lemma 21]{BL23} entails that a duo-chain can be computed in non-deterministic logarithmic space, thus providing a non-deterministic logspace proof for computing the Wagner hierarchy on normalized FDFAs.

\begin{restatable}{proposition}{thmwagneronnormalizedfdfas}\label{thm:wagner-on-normalized-fdfas}
    The Wanger position of an FDFA using normalization can be computed in NLOGSPACE.
\end{restatable}

\subsection{Succinctness Results}\label{sec:succinctness}

It is shown in \cite{BL23} (\cite{BohnL22}) that the minimal DPA (DBA) can be exponentially more succinct than the precise family of weak priority mappings from which it follows that it can be exponentially more succinct than the colorful FDFA. As the following proposition shows, the lack of succinctness comes from the colorful FDFA requirements, not the duo-normalized acceptance condition.  Intuitively, the reason is that the colorful FDFA insists to give the correct color to every finite period, whereas to recognize the language it suffices to give the correct color only to infinite words, and being imprecise regarding certain finite periods, allows merging states corresponding to different equivalence classes of $\progequiv{}{\colorful}$ and thus can potentially lead to smaller progress DFAs. \autoref{thm:colorful-exp-smaller-syntactic} shows that this can lead to an exponential gain. 
We first demonstrate this with an example.
\begin{example}\label{example:a-and-bb-comparison}
Consider the language $\infty a \wedge \infty bb$. The colorful FDFA for it has a one state leading automaton and the progress DFA
    $\prog{}{\colorful}$ given in \autoref{fig:a-and-bb}.
    As a comparison, the progress DFA of the syntactic/recurrent/limit FDFA (which are the same for this example) is given by $\prog{}{\syntactic}$.
    Next to them there is the progress DFA $\prog{}{\minimlduo}$ of (a 
    one-state leading automaton)  duo-normalized FDFA that also recognizes this language.
    We can see that the syntatic FDFA has $8$ states, the colorful has $6$ states and $\aut{F}^\minimlduo$ has only $4$ states.
\end{example}

\begin{figure}[t]
\begin{center}
\scalebox{0.5}{
    \begin{tikzpicture}[->,>=stealth',shorten >=1pt,auto,node distance=2.0cm,semithick,initial text=,  initial above]
    
    \node[state,initial]    (ppp1)     {$\stateppp{1}$};
    \node[state]    (ppp2) [below of=ppp1]       {$\stateppp{2}$};
    \node[state]    (ppp3) [below of=ppp2]       {$\stateppp{3}$};
    \node[state,accepting]    (ppp4) [below of=ppp3]       {$\stateppp{4}$};
    \node[state]    (ppp5) [right of=ppp2]       {$\stateppp{5}$};
    \node[state]    (ppp6) [right of=ppp3]       {$\stateppp{6}$};
    \node[state]    (ppp7) [right of=ppp5]       {$\stateppp{7}$};
    \node[state,accepting]    (ppp8) [right of=ppp6]       {$\stateppp{8}$};
    \node[label]            (labelPe) [above of=ppp1, node distance=1.5cm] {$\prog{}{\syntactic}:$};

    \node[state,initial]    (pp1) [right of=ppp1, node distance=8.5cm]          {$\statepp{1}$};
    \node[state]    (pp2) [below of=pp1]       {$\statepp{2}$};
    \node[state]    (pp3) [below of=pp2]       {$\statepp{3}$};
    \node[state,accepting]    (pp4) [below of=pp3]       {$\statepp{4}$};
    \node[state]    (pp5) [right of=pp2]       {$\statepp{5}$};
    \node[state]    (pp6) [right of=pp3]       {$\statepp{6}$};
    \node[label]            (labelPe) [above of=pp1, node distance=1.5cm] {$\prog{}{\colorful}:$};

    \node[state,initial]    (p1) [right of=pp1, node distance=7cm]     {$\statep{1}$};
    \node[state]    (p2) [below of=p1]       {$\statep{2}$};
    \node[state]    (p3) [below of=p2]       {$\statep{3}$};
    \node[state,accepting]    (p4) [below of=p3]       {$\statep{4}$};
    \node[label]            (labelPe) [above of=p1, node distance=1.5cm] {$\prog{}{\minimlduo}:$};

    \path (ppp1) edge               node {$a$} (ppp2);
    \path (ppp2) edge [loop left]              node {$a$} (ppp2);
    \path (ppp2) edge [bend left]             node {$b$} (ppp3);
    \path (ppp3) edge [bend left]             node {$a$} (ppp2);
    \path (ppp3) edge             node {$b$} (ppp4);
    \path (ppp4) edge [loop left]    node {$a, b$} (ppp4);
    \path (ppp1) edge               node {$b$} (ppp5);
    \path (ppp5) edge               node {$b$} (ppp6);
    \path (ppp6) edge [loop left]              node {$b$} (ppp6);
    \path (ppp6) edge               node {$a$} (ppp4);
    \path (ppp5) edge               node {$a$} (ppp7);
    \path (ppp7) edge [loop above]              node {$a$} (ppp7);
    \path (ppp7) edge [bend left]             node {$b$} (ppp8);
    \path (ppp8) edge [bend left]             node {$a$} (ppp7);
    \path (ppp8) edge             node {$b$} (ppp4);

    \path (pp1) edge               node {$a$} (pp2);
    \path (pp2) edge [loop left]              node {$a$} (pp2);
    \path (pp2) edge [bend left]             node {$b$} (pp3);
    \path (pp3) edge [bend left]             node {$a$} (pp2);
    \path (pp3) edge             node {$b$} (pp4);
    \path (pp4) edge [loop left]    node {$a, b$} (pp4);
    \path (pp1) edge               node {$b$} (pp5);
    \path (pp5) edge               node {$a$} (pp2);
    \path (pp5) edge               node {$b$} (pp6);
    \path (pp6) edge [loop left]              node {$b$} (pp6);
    \path (pp6) edge               node {$a$} (pp4);
 
    \path (p1) edge [loop right]    node {$b$} (p1);
    \path (p1) edge               node {$a$} (p2);
    \path (p2) edge [loop left]              node {$a$} (p2);
    \path (p2) edge [bend left]             node {$b$} (p3);
    \path (p3) edge [bend left]             node {$a$} (p2);
    \path (p3) edge             node {$b$} (p4);
    \path (p4) edge [loop left]    node {$a, b$} (p4);
   \end{tikzpicture}}
\end{center}	
\caption{Three progress DFAs $\prog{}{\minimlduo}$, $\prog{}{\colorful}$ and $\prog{}{\syntactic}$.}\label{fig:a-and-bb}
\end{figure}

\begin{restatable}{theorem}{thmcolorfulexpsmallersyntactic}\label{thm:colorful-exp-smaller-syntactic}
    The colorful FDFA can be exponentially more succinct than the syntactic FDFA.
\end{restatable}
The proof uses the family of languages $\{L_n\}_{n\in\mathbb{N}}$ over $\Sigma=\{a, b, \langle, \rangle\}$ where $L_n$ accepts all words with infinitely many occurrences of $\langle a^k b^m \rangle$ for some $1 \leq k \leq n$ and $m$ that is divisible by the $k$-th prime. The idea is that the colorful FDFA, since it uses duo normalization, can look only for prefixes of the form $\langle a^k b^m \rangle$, whereas the syntactic FDFA, since it uses normalization should answer correctly also for prefixes of the form  $b^m\rangle \langle a^k$ and recognizing such prefixes is much harder.

\begin{figure} 
\begin{center}
\scalebox{0.5}{
    \begin{tikzpicture}[->,>=stealth',shorten >=1pt,auto,node distance=2.6cm,semithick,initial text=,  initial above,
    g/.style={minimum size=1cm,fg=#1},
    fg/.code={\pgfmathtruncatemacro{\iFill}{100*#1}
        \tikzset{fill=black!\iFill}}]

    \coordinate (O) at (23,3);
    \coordinate (B) at (23,0);
    \coordinate (C) at (10,2.3);    
    \coordinate (CC) at (10,1.75);    
    \coordinate (CCC) at (10,1.4);    
    \coordinate (CCCC) at (10,1.2);    
    \coordinate (CCCCC) at (10,0.6);    
    \coordinate (P) at (13,1.35);    
    \coordinate (Pp) at (10.7,1.9);    
    \coordinate (Ps) at (11.7,1.7);    
    \coordinate (Pr) at (13.7,0.9);    
    \coordinate (Pl) at (13.7,1.7);    
    \coordinate (Pc) at (14.9,1.05);    
    \coordinate (Prp) at (13.2,1);    
    \coordinate (CCCCCd) at (7,-2);
    \coordinate (CCCCCn) at (7.45,-0.45);

    \draw (CCCCC)[line width=1mm, g=0.2 ] ellipse (6cm and 3cm);
    \draw (CCCC)[line width=1mm, g=0.1]  ellipse (4.2cm and 2.1cm);    
    \draw (C)[line width=1mm,g=0]  ellipse (1.4cm and .75cm);    
    \draw (CC)  ellipse (2cm and 1cm);
    \draw (CCC) ellipse (3cm and 1.5cm);
    \draw[rotate=350,dashed,blue] (P) ellipse (2.6cm and 1cm);
    \node[label,rotate=355] (proper)  [below of=Prp, node distance=5mm] {\figcls{\textcolor{blue}{Proper}}};     
    
    \node[label] (p)  [below of=Pp, node distance=0mm] {\figtxt{p}}; 
    \node[label] (s)  [below of=Ps, node distance=0mm] {\figtxt{s}}; 
    \node[label] (r)  [below of=Pr, node distance=0mm] {\figtxt{r}};
    \node[label] (l)  [below of=Pl, node distance=0mm] {\figtxt{l}};    
    \node[label] (c)  [below of=Pc, node distance=0mm] {\figtxt{c}};

    \path (CCCCCd) edge [double]  node {Thm.~\ref{thm-duo-succ}}  (CCCCCn);    

    \node[label] (exact)  [below of=C, node distance=.1mm] {\figcls{\Large\textbf{Exact}}}; 
    \node[label] (synct)  [below of=C, node distance=12mm] {\figcls{$\sim$-syntactic}};       
    \node[label] (proj)  [below of=C, node distance=20mm] {\figcls{Projective}};       
    \node[label] (norm)  [below of=C, node distance=27mm] {\figcls{\Large\textbf{Normalized}}};     
    \node[label] (norm)  [below of=C, node distance=40mm] {\figcls{\Large\textbf{Duo-normalized}}};       

    \node[label] (periodic)  [above of=O, node distance=2.2cm]                    {\figtext{periodic}};
    \node[label] (syntactic) [below of=periodic] {\figtext{syntactic}};
    \node[label] (recurrent) [below left of=syntactic] {$~$};
    \node[label] (recurrent-text) [ left of=recurrent, node distance=5mm] {\figtext{recur.}};
    \node[label] (limit)     [below right of=syntactic]       {$~$};
    \node[label] (limit-text)     [right of=limit, node distance=3mm]       {\figtext{limit}};
    \node[label] (colorful)  [below of=B]       {\figtext{colorful}};
    \node[label] (minimalduo) [below of=colorful] {\figtext{proper-duo}};

    \path (minimalduo) edge [->, double]   node [left] {Thm.~\ref{thm:minduo-exp-smaller-colorful}} (colorful);
    \path (colorful) edge [double] node  {Thm.~\ref{thm:colorful-exp-smaller-syntactic}} (recurrent);
    \path (colorful) edge [double] node [xshift=28pt,yshift=-10pt,{fill=white}] {\begin{tabular}{c}Thm.~\ref{thm:colorful-exp-smaller-syntactic}\end{tabular}} (syntactic);
    \path (colorful) edge [double] node [right] {Thm.~\ref{thm:colorful-exp-smaller-syntactic}} (limit);

    \path (recurrent) edge [<->,]                node [near end,yshift=-7pt,{fill=white}] {\cite{LST23}} (limit);
    \path (recurrent) edge   node [above, left] {\cite{AngluinF16}} (syntactic);	
    \path (limit) edge                 node [above, right] {\cite{LST23}} (syntactic);
    \path (syntactic) edge [double]  node {\cite{AngluinF16}} (periodic);

   \end{tikzpicture}}
   \caption{Left: The inclusions among these classes of FDFAs (in black), as well as the placement of the canonical FDFAs in these classes (in blue). The letters \textsc{p,s,r,l,c} abbreviate \textsc{periodic}, \textsc{syntactic},
   \textsc{recurrent},
   \textsc{limit}, and
   \textsc{colorful}, resp.
Right: Picture summarizing succinctness results on proper FDFAs. A double-line (resp. one-line) arrow form $\figtext{c}$ to $\figtext{d}$ indicates that 
$\figtext{c}$ can be exponentially (resp. quadratically) more succinct than $\figtext{d}$. 
}\label{fig:succ-results}\label{fig:class-incl}
\end{center}
\end{figure}
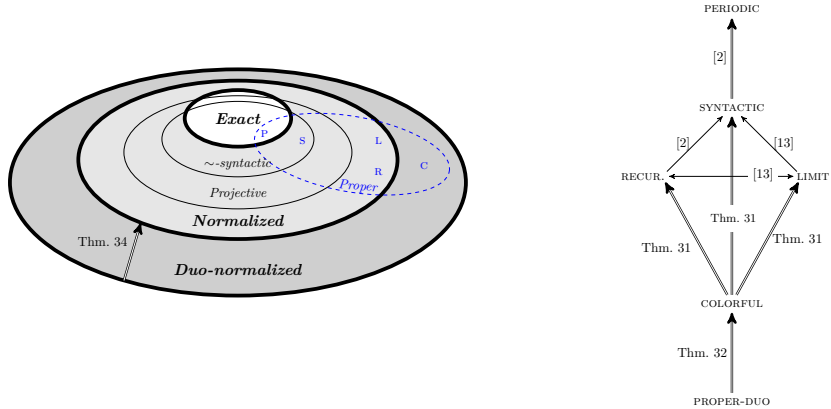

\begin{restatable}{theorem}{thmminduoexpsmallercolorful}\label{thm:minduo-exp-smaller-colorful}
 Duo-normalized FDFAs can be exponentially more succinct than the colorful FDFA.
\end{restatable}
The essence of the proof is that for languages whose requirement is that a set of infixes occur infinitely often,
a duo-normalized FDFA, similar to a DBA, can look for the infixes in an  arbitrary fixed order, whereas
the colorful FDFA must record at each time whether it has seen all infixes or not.  This is since until it hasn't seen all infixes the color is $1$ and once it has seen them all, the color is $0$. The proof thus uses the same family of languages as in~\cite[Prop. 14]{BohnL22}.

We note that all the canonical models (the periodic, syntactic, recurrent, limit and colorful) use $\aut{A}[\sim_L]$ for the leading automaton. 
An FDFA in general can use any leading automaton $\aut{A}[\sim]$ for a right congruence $\sim$ that refines $\sim_L$.
We term FDFAs in which the leading automaton is $\aut{A}[\sim_L]$ \emph{proper}. 
The FDFA used in the proof of \autoref{thm:minduo-exp-smaller-colorful} is proper. We can thus strengthen the claim as follows.
\begin{corollary}\label{cor:minduo-exp-smaller-colorful}
 Proper duo-normalized FDFAs can be exponentially more succinct than the colorful FDFA.
\end{corollary}

Surprisingly, Klarlund has shown that non-proper normalized FDFAs may be exponentially more succinct than proper normalized FDFAs~\cite{Klarlund94}.
One may thus wonder if non-proper normalized FDFAs can be as succinct as duo-FDFAs. That is, if duo-normalization adds succinctness when considering non-proper FDFAs. The following theorem shows that duo-normalization does adds succinctness even relative to non-proper FDFAs (and even when limiting duo-normalized FDFAs to proper ones).

\begin{restatable}{theorem}{thmduosucc}\label{thm-duo-succ}
(Proper) duo-normalized FDFAs can be exponentially more succinct than (not necessarily proper) normalized FDFAs. 
\end{restatable}
The proof uses the following family of languages
  over  $\Sigma = \Sigma_a \cup \Sigma_s $ where
     $\Sigma_a=\{a_1,\ldots,a_n\}$ and $\Sigma_s=\{s_1,\ldots,s_n\}$.
    $$L_n = \left\{ w\in(\Sigma^*\Sigma_a)^\omega \left|
    \begin{array}{cc}
        \text{Let }
        m = \max \{j ~|~  a_j \in \Sigma_a \cap \inf(w) \}.  \\
        \text{Then }  s_m\in\Sigma_s \text{ appears inf. often in } w.
    \end{array}\right.\right\}.$$
  
    The challenge in the language can be observed in periods where $s_m$ occurs before $a_m$ was seen (for $m$ being the maximal index of an $a_i$ letter in the period). The duo-normalized FDFA has the privilege of looking for duo-normalized decomposition in which
    $s_m$ is observed after the maximal $a_m$ is seen, whereas it is more complicated for a normalized FDFA (see appendix).

\section{Discussion}\label{sec:conclusions}
We provided a measure on FDFAs that corresponds to the Wagner hierarchy, we showed that it corresponds to the semantic notion of a reliable chain, and that a reliable chain in normalized or duo-normalized FDFAs implies the syntactic notion of a persistent chain. For the other direction, going from normalized FDFAs a duo-normalized chain implies a reliable chain whereas going from duo-normalized FDFAs one needs to look for persistent chains, and answering whether a duo-normalized FDFA has a persistent chain of length $k$ is PSPACE-complete.

With regard to succinctness, we have shown that FDFA with duo-normalized acceptance can be exponentially more succinct than FDFAs using (standard) normalization.
At the same time the common operations procedures and decision problems on them can still be done in NLOGSPACE. It is thus interesting to consider them as an acceptor for $\omega$-regular languages in verification (model-checking). \autoref{fig:succ-results} (right side)  summarizes the results regarding succinctness among the canonical FDFAs suggested thus far. It shows that the colorful FDFA can be exponentially more succinct than all other canonical models. At the same time, a minimal duo-normalized FDFA can be exponentially more succinct than the colorful FDFA. 

The figure might raise the question whether some duo-normalized FDFA can be doubly-exponentially more succinct than the periodic FDFA (the least succinct canonical representation). However this cannot be since a duo-normalized FDFA can be translated into an non-deterministic B\"uchi automaton (NBA) using exactly the same procedure as the one transforming a normalized FDFA into an NBA~\cite{AngluinBF18}. The reason is that the construction actually looks for a duo-normalized decomposition (which by saturation exists). 

\begin{proposition}
    If $L$ has a 
    duo-normalized FDFA $\aut{F}$ then it has an NBA of size polynomial in the number of states of $\aut{F}$.
\end{proposition}
Since an NBA can be converted into a periodic FDFA in an exponential blowup~\cite{CalbrixNP93old,KuperbergPP21} we get an overall exponential translation from duo-normalized FDFAs to the periodic FDFA, showing no doubly-exponential lower bound can be achieved. 
Since NBAs can be converted to DPAs with an exponential blow up~\cite{Safra88,Piterman06,Schewe09,FismanL15} and DPAs can be polynomially converted into non-proper FDFAs~\cite{AngluinBF18} we can conclude the following.
\begin{corollary}
The periodic FDFA, (non-proper) normalized FDFAs and DPAs of a language $L$ are at most exponentially larger than a duo-normalized FDFA for $L$.
\end{corollary}

\quad\\   

\textbf{Ackowledgments}
We thank Le\'on Bohn and Christof L\"oding for an interesting conversation on the subject.

\bibliographystyle{plainurl}
\bibliography{bib.bib}

\begin{thebibliography}{10}

\bibitem{AngluinBF18}
Dana Angluin, Udi Boker, and Dana Fisman.
\newblock Families of dfas as acceptors of {\(\omega\)}-regular languages.
\newblock {\em Log. Methods Comput. Sci.}, 14(1), 2018.

\bibitem{AngluinF16}
Dana Angluin and Dana Fisman.
\newblock Learning regular omega languages.
\newblock {\em Theor. Comput. Sci.}, 650:57--72, 2016.

\bibitem{BohnL22}
Le{\'{o}}n Bohn and Christof L{\"{o}}ding.
\newblock Passive learning of deterministic b{\"{u}}chi automata by combinations of dfas.
\newblock In {\em 49th International Colloquium on Automata, Languages, and Programming, {ICALP} 2022, July 4-8, 2022, Paris, France}, pages 114:1--114:20, 2022.

\bibitem{BL23}
Le{\'{o}}n Bohn and Christof L{\"{o}}ding.
\newblock Constructing deterministic parity automata from positive and negative examples.
\newblock {\em CoRR}, abs/2302.11043, 2023.

\bibitem{CalbrixNP93old}
H.~Calbrix, M.~Nivat, and A.~Podelski.
\newblock Ultimately periodic words of rational {\it w}-languages.
\newblock In {\em 9th Inter. Conf. on Mathematical Foundations of Programming Semantics (MFPS)}, pages 554--566, 1993.

\bibitem{CartonM99}
Olivier Carton and Ram{\'{o}}n Maceiras.
\newblock Computing the rabin index of a parity automaton.
\newblock {\em {RAIRO} Theor. Informatics Appl.}, 33(6):495--506, 1999.
\newblock \href {https://doi.org/10.1051/ita:1999129} {\path{doi:10.1051/ita:1999129}}.

\bibitem{DziembowskiJW97}
Stefan Dziembowski, Marcin Jurdzinski, and Igor Walukiewicz.
\newblock How much memory is needed to win infinite games?
\newblock In {\em Proceedings, 12th Annual {IEEE} Symposium on Logic in Computer Science, Warsaw, Poland, June 29 - July 2, 1997}, pages 99--110, 1997.

\bibitem{EhlersS22}
R{\"{u}}diger Ehlers and Sven Schewe.
\newblock Natural colors of infinite words.
\newblock In {\em 42nd {IARCS} Annual Conference on Foundations of Software Technology and Theoretical Computer Science, {FSTTCS} 2022, December 18-20, 2022, {IIT} Madras, Chennai, India}, pages 36:1--36:17, 2022.
\newblock \href {https://doi.org/10.4230/LIPIcs.FSTTCS.2022.36} {\path{doi:10.4230/LIPIcs.FSTTCS.2022.36}}.

\bibitem{FismanL15}
Dana Fisman and Yoad Lustig.
\newblock A modular approach for {B}{\"{u}}chi determinization.
\newblock In {\em 26th International Conference on Concurrency Theory, {CONCUR} 2015, Madrid, Spain, September 1.4, 2015}, pages 368--382, 2015.

\bibitem{GurevichH82}
Yuri Gurevich and Leo Harrington.
\newblock Trees, automata, and games.
\newblock In {\em Proceedings of the 14th Annual {ACM} Symposium on Theory of Computing, May 5-7, 1982, San Francisco, California, {USA}}, pages 60--65, 1982.

\bibitem{Klarlund94}
Nils Klarlund.
\newblock A homomorphism concepts for omega-regularity.
\newblock In {\em Computer Science Logic, 8th International Workshop, {CSL} '94, Kazimierz, Poland, September 25-30, 1994, Selected Papers}, pages 471--485, 1994.

\bibitem{KuperbergPP21}
Denis Kuperberg, Laureline Pinault, and Damien Pous.
\newblock Coinductive algorithms for b{\"{u}}chi automata.
\newblock {\em Fundam. Informaticae}, 180(4):351--373, 2021.

\bibitem{LST23}
Yong Li, Sven Schewe, and Qiyi Tang.
\newblock A novel family of finite automata for recognizing and learning $\omega$-regular languages, 2023.
\newblock ATVA 2023.
\newblock \href {https://arxiv.org/abs/2307.07490} {\path{arXiv:2307.07490}}.

\bibitem{li2023novel}
Yong Li, Sven Schewe, and Qiyi Tang.
\newblock A novel family of finite automata for recognizing and learning $\omega$-regular languages, 2023.
\newblock ATVA 2023.
\newblock \href {https://arxiv.org/abs/2307.07490} {\path{arXiv:2307.07490}}.

\bibitem{LiSTCX19}
Yong Li, Xuechao Sun, Andrea Turrini, Yu{-}Fang Chen, and Junnan Xu.
\newblock {ROLL} 1.0: {\textbackslash}omega -regular language learning library.
\newblock In {\em Tools and Algorithms for the Construction and Analysis of Systems - 25th International Conference, {TACAS} 2019, Held as Part of the European Joint Conferences on Theory and Practice of Software, {ETAPS} 2019, Prague, Czech Republic, April 6-11, 2019, Proceedings, Part {I}}, pages 365--371, 2019.

\bibitem{MalerS97}
Oded Maler and Ludwig Staiger.
\newblock On syntactic congruences for omega-languages.
\newblock {\em Theor. Comput. Sci.}, 183(1):93--112, 1997.

\bibitem{Myhill57}
J.~Myhill.
\newblock Finite automata and the representation of events.
\newblock Technical report, Wright Patterson AFB, Ohio, 1957.

\bibitem{Nerode58}
A.~Nerode.
\newblock Linear automaton transformations.
\newblock In {\em Proceedings of the American Mathematical Society, 9(4)}, page 541–544, 1958.

\bibitem{PerrinPinBook}
Dominique Perrin and Jean{-}Eric Pin.
\newblock {\em Infinite words - automata, semigroups, logic and games}, volume 141 of {\em Pure and applied mathematics series}.
\newblock Elsevier Morgan Kaufmann, 2004.

\bibitem{Piterman06}
Nir Piterman.
\newblock From nondeterministic buchi and streett automata to deterministic parity automata.
\newblock In {\em 21th {IEEE} Symposium on Logic in Computer Science {(LICS} 2006), 12-15 August 2006, Seattle, WA, USA, Proceedings}, pages 255--264, 2006.

\bibitem{Buchi62}
{B\"{u}chi}~J. R.
\newblock On a decision method in restricted second order arithmetic.
\newblock In {\em Int. Congress on Logic, Method, and Philosophy of Science}, pages 1--12. Stanford University Press, 1962.

\bibitem{Safra88}
Shmuel Safra.
\newblock On the complexity of omega-automata.
\newblock In {\em 29th Annual Symposium on Foundations of Computer Science, White Plains, New York, USA, 24-26 October 1988}, pages 319--327, 1988.

\bibitem{Schewe09}
Sven Schewe.
\newblock B{\"{u}}chi complementation made tight.
\newblock In {\em 26th International Symposium on Theoretical Aspects of Computer Science, {STACS} 2009, February 26-28, 2009, Freiburg, Germany, Proceedings}, pages 661--672, 2009.

\bibitem{Wagner75}
Klaus~W. Wagner.
\newblock A hierarchy of regular sequence sets.
\newblock In {\em Mathematical Foundations of Computer Science 1975, 4th Symposium, Mari{\'{a}}nsk{\'{e}} L{\'{a}}zne, Czechoslovakia, September 1-5, 1975, Proceedings}, pages 445--449, 1975.
\newblock \href {https://doi.org/10.1007/3-540-07389-2_231} {\path{doi:10.1007/3-540-07389-2_231}}.

\bibitem{WilkeY96}
Thomas Wilke and Haiseung Yoo.
\newblock Computing the rabin index of a regular language of infinite words.
\newblock {\em Inf. Comput.}, 130(1):61--70, 1996.

\end{thebibliography}

\appendix

\section{Omitted Proofs}
\clmduoexists*

\begin{claim}\label{clm:perst-exists}
For every $x \in\Sigma^*$ and $y \in\Sigma^+$ the word $xy^\omega$ has a persistent decomposition of the form $(xy^i,y^j)$ where $i$ and $j$ are of size quadratic in $\aut{F}$.
\end{claim}

We prove \autoref{clm:duo-exists} and \autoref{clm:perst-exists} together.

\begin{proof}
 Let $\aut{F}=(\aut{Q},\{\aut{P}\})$ be an FDFA.  Let $m=|\aut{Q}|$ and $n=|\prog{u}{}|$. 
 Let $u\in\Sigma^*$ and $v\in\Sigma^*$.
 For exact normalization the claim clearly holds for $i=0$ and $j=1$.

 For (standard) normalization, we require that $(u,v)$ is strongly normalized wrt $\aut{Q}$.
 Reading $v^\omega$ from any state of $\prog{u}{}$ we must reach a loop of at most $m$ states. Reading $v^m$ we can surely reach such a state. Thus there exists $0 \leq i\leq m$ such that $v^i$ reaches the loop.
 There are at most $m$ states in the loop, thus there exists a $1\leq j\leq m$ such that $v^j$ closes a loop on $\autstate{Q}{v^i}$.

 For duo-normalization, let $i$ and $j$ be such that $(uv^i,v^j)$ is normalized. Additionally, we require the period to be strongly-normalized wrt $\prog{u}{}$. From similar reasoning,  
 there exists $1 \leq k\leq 2n$ such that $v^k$ reaches and closes the final loop on $v$ in $\prog{u}{}$. 
 Let $l = \lcm(j,k)$ then $l\leq 2mn$ and $(uv^i,v^l)$ is duo-normalized.~\footnote{Note that a closer analysis may produce smaller $l$'s, such as taking the smallest multiple of $\lcm(j,k)$ that is greater than $n$, for $k$ that only closes the loop in $\prog{u}{}$.}

 For persistent normalization, we additionally require that, from every state in $\prog{u}{}$, after reading $v$ we reach the final loop on $v$. Again, since it takes at most $n$ states to reach the final loop (from any state) the same analysis as for duo-normalization works here.
\end{proof}

\clmduonormmainainsallgood*
\begin{proof}
Let $\aut{F} = (\aut{Q}, \aut{P})$ and $\aut{F}' = (\aut{Q}', \aut{P}')$ be saturated FDFAs using duo-normalization.

\emph{Complementation.} Let $\aut{F}^c = (\aut{Q}, \aut{P}^c)$ where $\aut{P}^c$ is the DFA obtained from $\aut{P}$ by flipping acceptance of states. Clearly, $(u,v)$ is duo-normalized in $\aut{F}$ iff it is duo-normalized in $\aut{F}^c$, as duo-normalization does not take into account accepting states. Therefore, since $\aut{F}^c$ will always reply to the contrary of $\aut{F}$, we can determine that $\sema{\aut{F}^c} = \overline{L}$.

\emph{Intersection.} Let $\aut{F} \otimes \aut{F}' = (\aut{Q} \times \aut{Q}', \{\aut{P} \otimes \aut{P}'\})$ be the FDFA obtained by looking at the cartesian multiplication of the leading and progress automata, and denoting a state $(q, q')$ accepting iff both $q$ and $q'$ are accepting in their respective progress DFAs. If a word $(u, v)$ is duo-normalized with respect to $\aut{F} \otimes \aut{F}'$, then $Q \times Q'(u) = Q \times Q' (u \cdot v)$ and $\prog{\aut{Q} \times \aut{Q}'(u)}{}(v) = \prog{\aut{Q} \times \aut{Q}'(u)}{}(v^2)$, so therefore $\aut{Q}(u) = \aut{Q}(u \cdot v), \aut{Q}'(u) = \aut{Q}'(u \cdot v)$ and $\prog{\aut{Q}(u)}{}(v) = \prog{\aut{Q}(u)}{}(v^2), \prog{\aut{Q'}(u)}{}(v) = \prog{\aut{Q'}(u)}{}(v^2)$, and therefore $(u,v)$ is duo-normalized with respect to $\aut{F}$ and $\aut{F}'$. Hence, since $\aut{F}$ and $\aut{F}'$ are saturated, and a duo-normalized word with respect to $\aut{F} \otimes \aut{F}'$ is also duo-normalized with respect to them, we deduce that $\sema{\aut{F} \otimes \aut{F}'} = \sema{\aut{F}} \cap \sema{\aut{F}'}$. Note that $\aut{F} \otimes \aut{F}'$ can be generated by iterating over the states of $\aut{F}$ and $\aut{F}'$, which requires space logarithmic in their sizes.

\emph{Union.} Let $\aut{F} \oplus \aut{F}' = (\aut{Q} \times \aut{Q}', \{\aut{P} \oplus \aut{P}'\})$ be the FDFA obtained by looking at the cartesian multiplication of the leading and progress automata, and denoting a state $(q, q')$ accepting iff either $q$ or $q'$ are accepting in their respective progress DFAs. By identical reasoning to the intersection case, we deduce that $\sema{\aut{F} \oplus \aut{F}'} = \sema{\aut{F}} \cup \sema{\aut{F}'}$.

\emph{Membership.} 
Given a pair $(u,v)$, by \autoref{clm:duo-exists} there exists $i,j$ quadratic in the size of $\aut{F}$ such that $(uv^i,v^j)$ is duo-normalized. The question of membership can be answered by running $(uv^i,v^j)$ in $\aut{F}$. Memory-wise, because of the limitation on the size of $i$ and $j$, this can be computed in logarithmic space.

\emph{Emptiness.} 
We are asking whether there exists a duo-normalized pair $(x,y)$ accepted by $\aut{F}$. We can guess $x$ (one letter at a time, at most $|\aut{Q}|$ letters) and traverse the leading automaton $\aut{Q}$ to reach some state $q$. We then guess an accepting state $q'$ in the progress automaton $\prog{q}{}$ and a word $y$ (one letter at the time). While guessing $y$ we traverse $\autfromq{Q}{q}$, $\prog{q}{}$ and $\aut{P}_{q_{[q'}}$,
if the word is duo-normalized we should reach states $q$, $q'$ and $q'$ again, respectively. The length of $y$ need not be above $|\aut{Q}|\cdot|\prog{q}{}|^2$.   

\emph{Universality.} Trivially reducible to complementation and emptiness.

\emph{Containment and Equality. } $\sema{\aut{F}} \subseteq \sema{\aut{F}'}$ iff $\sema{\aut{F}} \cap \sema{{\aut{F}'}^c} = \emptyset$, therefore it is reducible to intersection, complementation and emptiness. Equality is trivially reducible to containment.
\end{proof}

\clmpersaccunnecessary*
\begin{proof}
    Let $\aut{F}=(\aut{Q},\{\aut{P}\})$ be an FDFA. 
    Let $\aut{F}_p$ be $\aut{F}$ using persistent-acceptance and $\aut{F}_d$ be $\aut{F}$ using duo-normalized acceptance. Assume towards contradiction that there exists some $(x,y)\in\sema{\aut{F}_d}$ for which $(x,y)\not\in\sema{\aut{F}_p}$ or vice versa.  
    Since persistent normalization implies duo-normalization, if
    $(x,y)$ is persistent-normalized then
    $(x,y)\in\sema{\aut{F}_p}$ iff $(x,y)\in\sema{\aut{F}_d}$.
    Assume $(x,y)$ is duo-normalized but not persistent-normalized.     
    There exists some $j \geq 1$ for which $(x,y' = y^j)$ is persistent.
    Assume $\aut{P}_x(y)=q$.
    Since $(x,y)$ is duo-normalized, $y$ loops on $q$. Thus, $\aut{P}_x(y'){=}q$. Hence $(x,y)\in\sema{\aut{F}_d}$ iff $q$ is accepting iff $(x,y')\in\sema{\aut{F}}_p$. 
\end{proof}

\clmperspump*
\begin{proof}
    Let $q_1,q_2,\ldots,q_n$ be the states of $\aut{A}$. 
    Let $q'_1,q'_2,\ldots,q'_n$ be the states reached from $q_1,q_2,\ldots,q_n$ after reading $v$, respectively. 
    As $v$ is persistent-normalized, for every state $q'_i$ reading enough repetitions of $v$ loops back to $q'_i$. 
    Let $k_1,k_2,\ldots,k_n$ be the number of repetitions of $v$ required to loop on $q'_1,q'_2,\ldots,q'_n$ respectively. Let $k$ be the multiplication of all the $k_i$'s. It follows that $v^k$ closes a loop on all the $q'_i$'s. 
    Thus reading $v^k$, or $v^{ki}$ for every $i \geq 1$ and from every state, after reading $v$, does not change the reached state. Hence, if $vz$ is persistent (duo-normalized) then $v^{1+ki}z$ is persistent (duo-normalized) as well.
\end{proof}

\thmcompdiametermeas*
\begin{proof}
    Let $\aut{F}=(\aut{Q},\{\aut{P}\})$ be an FDFA. Let $u\in\Sigma^*$. 
    With every progress DFA $\prog{u}{}$ one can associate a DFA $\perstaut{u}$, which we term the \emph{persistent DFA}, that satisfies the following. Every state can be classified into \emph{significant} and \emph{insignificant}, such that 
    a word is persistent iff it reaches a significant state. The states of $\perstaut{u}$ are vectors of size $|\prog{u}{}|+1$ where the first entry ranges between $1$ and $|\aut{Q}|$ and the rest of the entries range between $1$ and $|\prog{u}{}|$. 
    Let $n=|\prog{u}{}|$.
    The initial state is $(\autstate{Q}{u},1,2,\ldots,n)$.
    There is an edge from state $(j,i_1,i_2,\ldots,i_n)$ to $(j',i'_1,i'_2,\ldots,i'_n)$ on $\sigma$
    if $\autfromq{\aut{Q}}{j}(\sigma)=j'$ and
    $\autfromq{\prog{u}{}}{i_k}(\sigma)=i'_k$ for every $1\leq k \leq n$.
    
    State $(j,i_1,i_2,\ldots,i_n)$ is accepting if $i_1$ is accepting in $\prog{u}{}$. 
    State $(j,i_1,i_2,\ldots,i_n)$ is \emph{significant} if the following conditions hold (i) $j=\autstate{Q}{u}$ (ii) $i_{i_1}=i_1$ and (iii) the indices $i_1,i_2,\ldots,i_n$ are a permutation of $1$ to $n$. Let $v$ be a word reaching a significant state. Observe that (i) entails that $v$ is $u$-invariant, (ii) entails that $v$ is strong-normalized wrt $\prog{u}{}$, and (iii) entails that 
    $v$ is weak-normalized wrt $\autfromq{\prog{u}{}}{q}$ for every $q$ in $\prog{u}{}$.~\footnote{This is since, the fact that indices form a permutation means we can detect \emph{cycles} in the vector, in the following sense. Let $K=\{k_1,\ldots,k_h\}$ be a subset of $\{1,\ldots,n\}$ we say that $K$ is a \emph{cycle} in $(j,i_1,i_2,\ldots,i_n)$ if $i_{k_m}=k_{m+1}$ for every $1\leq m < h$ and $i_{k_h}=i_{k_1}$. Such a cycle of length $h$ essentially guarantees that reading $v^h$ from every state in the cycle loops back to it. Thus after reading $v$ once we reach its terminal cycle as required.} 
    The three requirements implies that every word $v$ that reaches a significant state is persistent, and the same reasoning guarantees that if $v$ is persistent it reaches a significant state.

    Note that if $v$ reaches a significant-accepting state then $uv^\omega\in\sema{\aut{F}}$, and if $v$ reaches a significant-rejecting state then $uv^\omega\notin\sema{\aut{F}}$.
    From this and the previous claim we conclude that a path in $\perstaut{u}$ that has $k-1$ alternations between significant accepting and rejecting states corresponds to a persistent chain of length $k$, and vice versa.

    We provide a non-deterministic polynomial space algorithm to answer whether the positive (or negative) diameter measure of $\aut{F}$ is at least $k$. By Savitch's theorem there is also a deterministic polynomial space algorithm.
    The algorithm guesses words $u\in\Sigma^*$ and $v\in\Sigma^+$ letter by letter. The size of $u$ is bounded by $m$ and of $v$ by $m\cdot n^n$ where $m=|\aut{Q}|$ and $n=\max\{|\aut{P}|\}$. It simulates the run of $\perstaut{u}$ on $v$ recording the vector representing the state after reading the last letter, and counting the number of alternations between significant and insignificant states. (Note that the check if a state is  significant can be done in polynomial space.) 
\end{proof}

\propchainelmentexp*
\begin{proof}
    Let $p_1,p_2,\ldots,p_k$ be all the primes of size at most $n$.
    The proofs uses the construction of the FDFA in the proof of \autoref{thm:lb-diameter-meas} where the alphabet $\Sigma=\{a\}$
    and $D_i$ is a DFA accepting $(a^{p_i})^+$.
    Thus from the same arguments, the FDFA has a persistent chain of length $k+1$ and if $w_1 \prec w_2 \prec \cdots \prec w_{k+1}$ is a persistent chain then
    $w_{k+1}$ has a prefix of the form $x\sharp y$ where $y\in a^+$ and $y$ is in all $D_i$'s.
    From the nature of $D_i$, $y$ must be a multiplication of $\Pi_{1\leq i\leq m}p_i$.
    Since the number of primes up to $n$ is bounded by $\Theta(n/\log n)$ and since each prime if of size $2$ at least.
    The multiplication is of size at least $2^{\Theta(n/\log n)}$ while the number of states in the FDFA is quadratic in $n$.
\end{proof}

Below we state some useful properties of colors for finite words, that appear or can be distilled from the claims in~\cite{BL23}.
Since the formulation here defers from that in~\cite{BL23}, for completeness, we provide proofs using our notations.

\begin{proposition}\label{prop:finite-colors-properties}
Let $u,u',x\in\Sigma^*$, $v,y\in\Sigma^+$ and $w\in\Sigma^\omega$.
\begin{enumerate}[nosep]
    \item \label{clm:fin-clr-mono-ninc} 
    Monotonically non-increasing:
    If $\fincolorof{u}{v} = c$ then \emph{for all} $z\in\Sigma^*$ it holds that $\fincolorof{u}{vz} \leq c$.
    \item \label{clm:fin-clr-stp}
    Step by step:
    If $\fincolorof{u}{v} = c > 1$ then \emph{exists} $z\in\Sigma^*$ such that $vz$ is $u$-invariant and $\fincolorof{u}{vz} = c-1$.\footnote{Note that $1$ need not be the minimal non-negative color. We can later strengthen this and show (using \autoref{clm:fin-clr-diff-eqv}, see \autoref{clm:use-of-less-is-more-in-step-by-step}), that the claim holds for every $c> \mincoloru{u}$.}
    \item \label{clm:stbl-prd-evnt} 
    Every period eventually stabilizes:
    There {exists} $i > 0$ such that $v^i$ is stable wrt $u$. If $v$ is $u$-invariant then $v^i$ is also $u$-reliable.
    
    \item \label{clm:stbl-rlvnt-prd-acc}
    Reliable periods determine:
    If $v$ is $u$-reliable then $\fincolorof{u}{v}$ is even iff $uv^\omega\in L$.

    \item \label{clm:mon-decr} 
    Monotonically decreasing:
    Assume $v \prec vz$ are both $u$-reliable and $u(v)^\omega\in L$ iff $u(vz)^\omega\notin L$. Then $\fincolorof{u}{v} > \fincolorof{u}{vz}$.

    \item \label{clm:stbl-rlvnt-prd-ext} 
    Extending to reliable periods of the same color:
    If $v$ is $u$-relevant 
then \emph{exists} $z\in\Sigma^*$ s.t. $\fincolorof{u}{vz} 
=$ 
{$\fincolorof{u}{v}$} and $vz$ is $u$-reliable.

    \item \label{clm:fin-clr-diff-eqv}
    Less is more:
    If $u\sim_L xy$ then $\fincolorof{u}{v}\geq \fincolorof{x}{yv}$.
    
\end{enumerate}    
\end{proposition}

Below we use 
 $\mincolorof{L}$ for $\min\{ c\geq 0 ~\colon~\fincolorof{u}{v}=c,\ u\in\Sigma^*,v\in\Sigma^+\}$
 and
 $\mincoloru{u}$ to be $\min\{ c\geq 0 ~\colon~\fincolorof{u}{v}=c,\ v\in\Sigma^+\}$ if this set is non-empty and 
  $\mincolorof{L}$ otherwise.

\begin{proof}
\begin{enumerate}[nosep]
    
    \item  
    Monotonically non-increasing:
Assume towards contradiction that exists $z\in\Sigma^*$ such that $\fincolorof{u}{vz} > c$. If $c = \neginfty$ then it follows that exists $x\in\Sigma^*$ such that $vzx$ is $u$-invariant, contradicting $v$ being irrelevant.
Otherwise $c \geq 0$ and it follows that exists $x\in\Sigma^*$ such that $vzx$ is $u$-invariant for which $u(vzx)^\omega\not\in L$ iff $c$ is even, and for all $i > 0$ it holds that $\fincolorof{u}{(vzx)^i} \geq c$. It follows that $zx$ contradicts $\fincolorof{u}{v} = c$.

    \item 
    Step by step: 
If $\fincolorof{u}{v} = c > 1$ then from $\fincolorof{u}{v} > c-2 \geq 0$ it follows that exists $z\in\Sigma^*$ such that $vz$ is $u$-invariant for which $u(vz)^\omega\not\in L$ iff $c-2$ is even, and for all $i > 0$ it holds that $\fincolorof{u}{(vz)^i} \geq c-2$. From $\fincolorof{u}{v} = c$ it follows that $u(vz)^\omega\in L$ iff $c$ is even, or exists $i > 0$ such that $\fincolorof{u}{(vz)^i} < c$. Note that $c$ is even iff $c-2$ is even, thus there must exist $i > 0$ such that $\fincolorof{u}{(vz)^i} < c$. Because $\fincolorof{u}{(vz)^i} \geq c-2$ it follows that $\fincolorof{u}{(vz)^i} = c-1$. 
Hence $z' = z(vz)^{i-1}$ satisfies the claim.

    \item 
    Every period eventually stabilizes:
Assume towards contradiction that for all $i > 0$ it holds that $v^i$ is unstable wrt $u$. Then for every $i > 0$ exists $j > 0$ such that $\fincolorof{u}{v^i} \neq \fincolorof{u}{(v^i)^j}$ and from \autoref{clm:fin-clr-mono-ninc} it follows $\fincolorof{u}{v^i} > \fincolorof{u}{(v^i)^j}$. As the same is true for $i\times j$, exists $k > 0$ such that $\fincolorof{u}{v^{ij}} > \fincolorof{u}{(v^{ij})^k}$. By repeating the process we can create an infinitely decreasing chain of colors, contradicting the finite number of colors. If $v$ is $u$-invariant then so is $v^i$ and thus it is reliable.

    \item
    Reliable periods determine:
Wlog assume towards contradiction that $\fincolorof{u}{v} = c$ is even and $uv^\omega\not\in L$. Let $z=\epsilon$ then because $v$ is $u$-reliable it follows that $v = vz$ is $u$-invariant. From $\fincolorof{u}{v} = c$ and $c$ is even, it follows that either $u(vz)^\omega = uv^\omega\in L$ or exists $i > 0$ such that $\fincolorof{u}{(vz)^i} = \fincolorof{u}{v^i} < c$. From $uv^\omega\not\in L$, it follows that exists $i > 0$ such that $\fincolorof{u}{v^i} < c$. Because $v$ is $u$-reliable it holds that $\fincolorof{u}{v} = \fincolorof{u}{v^i} < c$ in contradiction.

    \item 
    Monotonically decreasing: 
    It follows from \autoref{clm:fin-clr-mono-ninc} that $\fincolorof{u}{v} \geq \fincolorof{u}{vz}$.
    It follows from \autoref{clm:stbl-rlvnt-prd-acc} that one color is even and the other one is odd. 
    Therefore $\fincolorof{u}{v} > \fincolorof{u}{vz}$

    \item 
    Extending to reliable periods of the same color:
Since $v$ is $u$-relevant we have $\fincolorof{u}{v}\geq 0$.
If $\fincolorof{u}{v} = 0$ then exists $z\in\Sigma^*$ such that $vz$ is $u$-invariant. From \autoref{clm:stbl-prd-evnt} exists $i > 0$ such that $(vz)^i$ is $u$-reliable. As $(vz)^i$ is relevant, from \autoref{clm:fin-clr-mono-ninc} we get $0 = \fincolorof{u}{v} \geq \fincolorof{u}{(vz)^i} > \neginfty$ as needed. Otherwise, $\fincolorof{u}{v} = c > c - 1 \geq 0$. 
It follows that exists $z\in\Sigma^*$ such that $vz$ is $u$-invariant for which $u(vz)^\omega\not\in L$ iff $c-1$ is even, and for all $i > 0$ it holds that $\fincolorof{u}{(vz)^i} \geq c-1$. 
As above, from \autoref{clm:stbl-prd-evnt} exists $j > 0$ such that $(vz)^j$ is stable and thus $u$-reliable. From \autoref{clm:fin-clr-mono-ninc}, $c = \fincolorof{u}{v} \geq \fincolorof{u}{(vz)^j} \geq c-1$. From \autoref{clm:stbl-rlvnt-prd-acc}, $\fincolorof{u}{(vz)^j} \neq c-1$ thus $\fincolorof{u}{(vz)^j} = c$ as needed.

    \item
    Less is more: 
    We prove the proposition by induction on the color $\fincolorof{u}{v}$.
We prove by induction on $c$, that $\fincolorof{u}{v}=c$ implies $\fincolorof{x}{yv}\leq c$, for all $v\in\Sigma^+$.

\begin{inparaitem}

    \item 
    In the base case $\fincolorof{u}{v} = \neginfty$. Assume towards contradiction that exists $z\in\Sigma^*$ such that $yvz$ is $x$-invariant. Then $x\sim_L xyvz$ thus $xy\sim_L xyvzy$ and from $xy\sim_L u$ we get $u\sim_L uvzy$, implying $vzy$ is  $u$-invariant, contradicting $v$ being irrelevant. Thus no such $z$ exists and $\fincolorof{x}{yv} = \neginfty \leq \fincolorof{u}{v}$.
    
    \item
    In the induction step $\fincolorof{u}{v} = c > \neginfty$ and the claim holds for every $c'<c$.
    Assume towards contradiction that $\fincolorof{x}{yv} > c$, by applying \autoref{clm:fin-clr-stp} repeatedly, we can decrease the color of $\fincolorof{x}{yv}$ step by step until reaching the color of $\fincolorof{u}{v}$ plus one, i.e., finding $z'\in\Sigma^*$ such that $\fincolorof{x}{yvz'} = c+1$ and $yvz'$ is $x$-invariant.
    
    From \autoref{clm:stbl-rlvnt-prd-ext} exists $z''\in\Sigma^*$ such that $\fincolorof{x}{yvz'z''} = c+1$ and $yvz'z''$ is $x$-reliable. Let $z = z'z''$, then $\fincolorof{x}{yvz} = c+1$ and because $yvz$ is $x$-invariant, we get that either $x(yvz)^\omega\in L$ iff $c+1$ is even or exists $i > 0$ such that $\fincolorof{x}{(yvz)^i} < c+1$. As $yvz$ is $x$-reliable it follows that for all $i > 0$ it holds that $c+1 = \fincolorof{x}{yvz} = \fincolorof{x}{(yvz)^i}$ thus it must hold that $x(yvz)^\omega\in L$ iff $c+1$ is even. Because $yvz$ is  $x$-invariant it holds that $x\sim_L xyvz$ thus $xy\sim_L xyvzy$, and from $xy\sim_L u$ it follows that $u\sim_L uvzy$ implying $vzy$ is $u$-invariant. From \autoref{clm:fin-clr-mono-ninc} it follows that $c = \fincolorof{u}{v} \geq \fincolorof{u}{vzy}$. 
    Thus, either $\fincolorof{u}{vzy}=c$ or $\fincolorof{u}{vzy}<c$.
    
    In the first case, $\fincolorof{u}{vzy} = c$. As $vzy$ is $u$-invariant we get that either $u(vzy)^\omega\in L$ iff $c$ is even or exists $i > 0$ such that $\fincolorof{u}{(vzy)^i} < c$. As $x(yvz)^\omega = xy(vzy)^\omega$ and $xy\sim_L u$ it follows that $x(yvz)^\omega\in L$ iff $u(vzy)^\omega\in L$. From $x(yvz)^\omega\in L$ iff $c+1$ is even it follows that $u(vzy)^\omega\in L$ iff $c+1$, contradicting $u(vzy)^\omega\in L$ iff $c$ is even. Thus must exist $i > 0$ such that $\fincolorof{u}{(vzy)^i} < c$. Considering now the period $vzy$ rather than $v$, using the induction hypothesis on $\fincolorof{u}{(vzy)^i} < c$ we get that $\fincolorof{u}{(vzy)^i} \geq \fincolorof{x}{y(vzy)^i}$. By restating $y(vzy)^i$ as $(yvz)^iy$, we get from \autoref{clm:fin-clr-mono-ninc} that $\fincolorof{x}{(yvz)^iy} \geq \fincolorof{x}{(yvz)^{i+1}}$. As $yvz$ is $x$-reliable it holds that $\fincolorof{x}{(yvz)^{i+1}} = \fincolorof{x}{yvz}$ and we get that $c > \fincolorof{x}{yvz} = c+1$, in contradiction. 
    
    In the second case it already holds that $\fincolorof{u}{vzy} < c$ thus from \autoref{clm:fin-clr-mono-ninc} it follows that $\fincolorof{u}{(vzy)^i} \leq \fincolorof{u}{vzy} < c$ which leads to the same contradiction. 
\end{inparaitem}
\end{enumerate}
\end{proof}

\begin{remark}\label{clm:use-of-less-is-more-in-step-by-step}
    We can now show that \autoref{clm:fin-clr-stp} also holds for $\fincolorof{u}{v} = 1$ given $\mincoloru{u} = 0$. As the color of $v$ is non-negative, it is relevant and thus there exists some $z\in\Sigma^*$ for which $vz$ is $u$-invariant. Because the minimal color of ${u}$ is $0$ there exists some $x\in\Sigma^*$ for which $\fincolorof{u}{x} = 0$. By applying \autoref{clm:fin-clr-diff-eqv} we get $0 \geq \fincolorof{u}{vzx}$ and as it is also $u$-invariant its color must be $0$ as needed.
\end{remark}

\begin{restatable}[Colors and Reliable Chains]{proposition}{clmcolorsandchains}\label{clm:colors-and-chains}
Let  $\mincoloru{u}=0$ (resp. $\mincoloru{u}=1$). Let $v$ be $u$-relevant.
Then 
$\fincolorof{u}{v}\geq c$ iff there exists a reliable chain 
$v_1 \prec v_2 \prec \ldots \prec v_k$  where $k=c+1$ (resp. $c$) such that $v\prec v_1$.
\end{restatable}

\begin{proof}
    Let $\fincolorof{u}{v}\geq c$. As $v$ is $u$-relevant $c > \neginfty$ and by repetitive applications of \autoref{clm:fin-clr-stp}
    and \autoref{clm:stbl-rlvnt-prd-ext} we get that there exists a reliable chain $v_1 \prec v_2 \prec \ldots \prec v_k$ where $\fincolorof{u}{v_1}=c$ and $\fincolorof{u}{v_{i+1}}= \fincolorof{u}{v_{i}}-1$. Thus if $\mincoloru{u}=0$ then the chain is of length $c+1$ and if it is $1$ then the length is $c$.  

    For the other direction, assume there exists a reliable chain of
$v_1 \prec v_2 \prec \ldots \prec v_k$ of
length $k$. It follows from repetitive applications of \autoref{clm:mon-decr} that 
the $\fincolorof{u}{v_1} \geq \fincolorof{u}{v_k}+k-1$.
Thus, if $\mincoloru{u}=0$ then $\fincolorof{u}{v_k}\geq 0$ and $\fincolorof{u}{v}\geq \fincolorof{u}{v_1} \geq c$.
Similarly for $\mincoloru{u}=1$.
\end{proof}

\clmpersimplrel*
\begin{proof}
    As there are only minor differences, we prove both claims simultaneously.
    In both cases $(u,v)$ is normalized thus $v$ is $u$-invariant and relevant.
    We show $v$ is stable by induction on $c = \finColorof{u}{v}$. As $v$ is relevant it holds that $\finColorof{u}{v^i} \geq 0$ for all $i \geq 1$. In the base case $c = 0$ and so $v$ must be stable. 
    
    For the induction step, assume towards contradiction that $v$ is not stable, thus from \autoref{clm:stbl-prd-evnt} exists $s \geq 1$ such that $v^s$ is stable and $\finColorof{u}{v^s} = c' < c = \finColorof{u}{v}$. 
    To later use \autoref{clm:pers-pump}, which is needed for the duo-normalized FDFA, we further increase the number of repetitions as follows,
    let $k$ be the number from \autoref{clm:pers-pump} and $l = 1 + ki$ for some $i$ such that $l \geq s$. In the normalized FDFA case we leave $l$ to be $s$.
    Then $v^l$ is also stable and of the same color (i.e. $c'$).
    From \autoref{clm:fin-clr-stp} exists $z_1 \in \Sigma^*$ such that $\finColorof{u}{vz_1} = c'+1$ and $vz_1$ is $u$-invariant. From \autoref{clm:stbl-rlvnt-prd-ext} exists $z_2 \in \Sigma^*$ such that $\finColorof{u}{vz_1z_2} = c'+1$ and $vz_1z_2$ is reliable. From \autoref{clm:duo-exists} exists $j \geq 1$ such that $(vz_1z_2)^j$ is $\prog{u}{}$-duo-normalized. Let $z$ be such that $vz = (vz_1z_2)^j$. As $vz_1z_2$ is reliable, it holds that $\finColorof{u}{vz} = c'+1$. and as it is $u$-invariant it follows that $(u,vz)$ is $\aut{Q}$-normalized, meaning $(u,vz)$ is duo-normalized in $\aut{F}$. 
    If $(u,v)$ is persistent then from \autoref{clm:pers-pump} and $(u,vz)$ being duo-normalized it follows that $(u,v^lz)$ is also duo-normalized in $\aut{F}$. 
    Even if $(u,v)$ is only duo-normalized, $(u,v^lz)$ reaches the same state as $(u,vz)$. This state is either accepting or rejecting and thus both words are either in the language or not (if $(u,v)$ is persistent then both are duo-normalized, otherwise both are normalized which is enough in the normalized FDFA case).
    Note that by induction hypothesis, both words are reliable, thus the parity of their colors must agree. This implies $c'' = \finColorof{u}{v^lz} \leq c' - 1 < c$. We can now repeat this process and extend $vz$ to some word $vzz'$ of color $c'' + 1$. To avoid contradiction, the color of $v^lzz'$ must decrease again. As the number of colors is finite, after a finite number of steps we ought to reach a contradiction.
\end{proof}

\clmrelichnimplpreschn*
\begin{proof}
    The proof is by induction on $d$.
    For $d=1$, by \autoref{clm:duo-exists} there exists an $i$ such that $y_1=(v_1)^i$ is persistent and it follows that $u(v_1)^\omega\in \sema{\aut{F}}$ iff $u(y_1)^\omega\in \sema{\aut{F}}$.

    For the induction step, let ${v_1 \prec v_2 \prec \ldots \prec v_d}$ be a reliable chain wrt $u$.
    Again, by \autoref{clm:duo-exists} there exists an $i$ such that $y_1=(v_1)^i$ is persistent. As $v_1$ is reliable it holds that $\finColorof{u}{v_1} = \finColorof{u}{y_1}$ and from \autoref{clm:colors-and-chains} there exists a reliable chain $y_1 \prec y_2 \prec y_3 \prec \ldots \prec y_d$. Both reliable chains start at the same color, thus ${u}(y_i)^\omega\in \sema{\aut{F}}$ iff ${u}(v_i)^\omega\in\sema{\aut{F}}$.
    
    By the induction hypothesis there exists a persistent chain $z_2 \prec z_3 \ldots \prec z_{d}$ such that
    $y_1 \prec y_2 \prec z_2$ and 
    ${u}(z_i)^\omega\in \sema{\aut{F}}$ iff ${u}(y_i)^\omega\in\sema{\aut{F}}$ for every $2\leq i \leq d$.
    By \autoref{clm:duo-exists} there exists a $j$ such that ${y_1}^j$ is persistent. Let $k$ be the number from \autoref{clm:pers-pump} and $l = j(1 + k)$.
    Let $z'_i$ be the word obtained from $z_i$ by replacing its $y_1$ prefix by $(y_1)^l$.
    Then by the lemma $z'_i$ are also persistent and reach the same state as $z_i$, thus ${u}(z'_i)^\omega\in \sema{\aut{F}}$ iff ${u}(z_i)^\omega\in\sema{\aut{F}}$.
    Let $z'_1=(y_1)^l$.
    Then $z'_1 \prec z'_2 \prec z'_3 \prec \ldots \prec z'_d$ is a persistent chain and 
    ${u}(z'_i)^\omega\in \sema{\aut{F}}$ iff ${u}(v_i)^\omega\in\sema{\aut{F}}$.
\end{proof}

\thmwagneronnormalizedfdfas*
\begin{proof}
Let $\aut{F}$ be an FDFA using normalized acceptance.
Following \autoref{clm:pers-impl-rel} and \autoref{clm:reli-chn-impl-pres-chn} the Wagner measure can be derived from the length of a maximal duo-normalized chain.
A non-deterministic Turing machine can guess the existence of such a chain and verify it as follows.
By~\cite[Lemma 21]{BL23}, in a projective FDFA if $x$ is duo-normalized wrt $\sim$ and $\approx_u$ then for every $y \approx_u x$ we have that $y$ is also duo-normalized.
Recall that any FDFA using normalization can be transformed with a quadratic blowup into a \emph{projective FDFA} by multiplying the progress DFAs by the leading automaton.
The machine guesses a duo-chain in a projective version of the FDFA (that it constructs on the fly).

It starts by guessing a word letter by letter until reaching some state $q$ of the leading automaton.  
It then guesses a number $k$ corresponding to the length of the chain in $\aut{P}_q$ (which is bounded by the number of states), and the parity $b$ of the first element in the chain.
Let $p_0$ be the initial state of $\aut{P}_q$ and set $p=p_0$.
For $i=1,2,\ldots,k$ it guesses a state $p'$ that is accepting iff $i\% 2=b$. 
It then guesses a word letter by letter and traces its run simultaneously from $q$ and $p$ until reaching $q$ and $p'$, resp.  
This proves that $(q,p')$ is reachable from $(q,p)$ in the projective FDFA version.
Then, to verify that there is a duo-normalized word reaching $(q,p')$ it guesses a word, letter by letter, from $(q,p_0,p')$ until reaching $(q,p',p')$ which
indicates that the word closes a loop in the leading automaton, and in the progress DFA of the projective FDFA it reaches $(q,p')$ from $(q,p_0)$ and closes a loop on it, which provides the desired duo-normalized witness.  
When this holds it updates $p$ to $p'$.

The memory it requires is for $k$, $b$, $q$, $p$, $p'$ and the current (pair or) tuple $(x_q,x_p,x_{p'})$, which are all of size logarithmic in the number of states of the FDFA.
Note that it suffices to bound the length of a word reaching $q$ by $n$, the number of states in the leading automaton, and the duo-normalized by $nm$ where $m$ is the size of leading automaton. 
\end{proof}

\thmcolorfulexpsmallersyntactic*
\begin{proof}
    Consider the family of languages $\{L_n\}_{n\in\mathbb{N}}$ over $\Sigma=\{a, b, \langle, \rangle\}$ where $L_n$ accepts all words with infinitely many occurrences of $\langle a^k b^m \rangle$ for some $1 \leq k \leq n$ and $m$ that is divisible by the $k$-th prime. We will show that the size of the colorful FDFA for $L_n$ is polynomial in $n$, and that the size of the syntactic FDFA is exponential.

    The language $L_n$ is prefix-independent, so the leading automation in both the colorful and syntactic FDFA consists of one state, and there is only one progress DFA.

    If $v \in \Sigma^*$ contains such $\langle a^k b^m \rangle$, then $\forall z \in \Sigma^* : (vz)^\omega \in L_n$, implying $\fincolorof{\epsilon}{v} = 0$. If on the other hand, $v$ contains no such $\langle a^k b^m \rangle$, then it is easy to see that $(v \langle \rangle)^\omega \notin L_n$, and hence $\fincolorof{\epsilon}{v} = 1$. Therefore, the progress DFA of the colorful FDFA accepts the language $\Sigma^* \langle a^k b^m \rangle \Sigma^*$ for some $1 \leq k \leq n$ and $m$ that is divisible by the $k$-th prime, which may be recognized by a DFA of size asymptotic to the sum of the first $n$ prime numbers, known to be polynomial in $n$.

    Contrasting this, the progress DFA of the syntactic FDFA must accept a word $ b^m \rangle\langle a^k$ iff $1 \leq k \leq n$ and $m$ is divisible by the $k$-th prime number. Hence, while reading {$b^m$} it must remember whether or not it is divisible by each of the first $n$ prime numbers, which requires at least as many states as the multiplication of the first $n$ prime numbers, known to be exponential in $n$.
\end{proof}

\thmminduoexpsmallercolorful*
\begin{proof}
    Consider the family of languages $\{L_n\}_{n \in \mathbb{N}}$ over $\Sigma = {a_1,a_2,...,a_n}$ where $L_n$ accepts all words with infinitely many occurrences of all letters $a_1,a_2,...,a_n$. We show there exists a duo-normalized FDFA for $L_n$ is at most linear in $n$, and that the size of the colorful FDFA is exponential.

    The language $L_n$ is prefix-independent, so the leading automaton in the colorful FDFA consists of a single state, and there is only one progress DFA.
    If $v \in \Sigma^*$ contains all letters $a_1,a_2,...,a_n$, then $\forall z \in \Sigma^* : (vz)^\omega \in L_n$, implying $\fincolorof{\epsilon}{v} = 0$. If on the other hand, $v$ does not contain all the desired letters, then $v^\omega \notin L_n$, and hence $\fincolorof{\epsilon}{v} = 1$. Therefore, the progress DFA of the colorful FDFA accepts the language $L = \{v \in \Sigma^* : v$ contains all letters $a_1,a_2,...,a_n\}$, whose minimal DFA must remember whether or not each letter was encountered, and therefore requires $2^n$ states.

    Contrasting this, we can build a duo-normalized FDFA $\aut{F}$ for $L_n$ as follows. The leading automaton of $\aut{F}$ has one state. The progress DFA $\prog{\epsilon}{}$ starts in a rejecting state searching for $a_1$, then if it is found transition to a rejecting state searching for $a_2$, and so on until once $a_n$ is found it'll transition to an accepting sink state. Note this DFA only uses $n$ states, and accepts the language $L = \Sigma^* a_1 \Sigma^* a_2 \cdots a_n \Sigma^*$.

    Let $v \in \Sigma^*$ be a duo-normalized word wrt. $\prog{\epsilon}{}$. If reading $v$ reaches the accepting sink, then by definition of $\prog{\epsilon}{}$ it must contain all letters $a_1,a_2,...,a_n$, and therefore $v^\omega \in L$. If on the other hand, reading $v$ reaches the rejecting state searching for $a_i$, then since reading $v$ loops on said state, we deduce $v$ does not contain $a_i$, and therefore that $v^\omega \notin L$.
    This proves the claim for the colorful FDFA.
\end{proof}

\thmduosucc*
\begin{proof}[Proof sketch]
    Let $\Sigma_a=\{a_1,\ldots,a_n\}$ and $\Sigma_s=\{s_1,\ldots,s_n\}$ and $\Sigma = \Sigma_a \cup \Sigma_s $.
    Consider the family of languages
    $$L_n = \left\{ w\in(\Sigma^*\Sigma_a)^\omega \left|
    \begin{array}{cc}
        \text{Let }
        m = \max \{j ~|~  a_j \in \Sigma_a \cap \inf(w) \}.  \\
        \text{Then }  s_m\in\Sigma_s \text{ appears inf. often in } w.
    \end{array}\right.\right\}.$$

    Let $\aut{F}$ be an FDFA for $L_n$.
    Consider $u\in\Sigma^*$ and let $\aut{Q}(u)=q_u$ where
    $\aut{Q}$ is the leading automaton. 
    Consider the progress DFA $\aut{P}_{q_u}$ and a word $v$ such that $\aut{Q}(u)=\aut{Q}(uv)$.
    When $\aut{P}_{q_u}$ reads $v$ it needs to accept iff $s_m$ occurs where $m = \max \{j ~|~  a_j \text{ occurs in } v\}$.
    Note that it could see $s_m$ before seeing $a_m$ where $m$ is the maximal index of a letter $a_j$ in $v$.
    Thus if $q_u$ does not convey any information about the maximal letter seen when looping on $q_u$, then $\aut{P}_{q_u}$ needs to track itself the maximal index $m$ of the letters $a_i$'s read and the set of all $s_j$ seen for $j\geq m$. Thus it needs $2^{n-m}$ states after seeing $a_m$, and before seeing any $a_i$ letters it needs $2^n$.
    So $\aut{Q}$ must convey some information for the FDFA to be of polynomial size.
    The most useful information would be the maximal index of a letter $a_i$ traversed on any loop on $q_u$. While $\aut{Q}$ does not know when it entered a period, it can keep track of this while reading increasing prefixes of $uv^\omega$ using the \textsc{lar} (latest appearance record) construction, but this requires $\Theta(n!)$ states~\cite{GurevichH82,DziembowskiJW97}.
    Suppose $q_u$ does convey some useful information (but not exactly $m$). Such useful information must be some restriction on the possible values of $m$. If such information is available then $\aut{P}_{q_u}$ needs to track only the relevant $s_j$'s.
    If the amount of possible such $m$'s is a function of $n$ then $\aut{P}_{q_u}$  still needs to track a subset of some function of $n$ which is exponential in $n$.
    Thus the leading automaton can only convey information  a fixed number of such $m$'s per state.
    If this was possible using less than exponentially many states, then so would detecting the true maximal among the set.

    Next we present a proper duo-normalized FDFA recognizing $L_n$ with $O(n)$ states.
    As $|\sim_{L_n}| = 1$ the leading automaton has a single state. The progress automaton has an initial state and $2n$ states $q_1,q'_1,q_2,q'_2,\ldots,q_n,q'_n$. The accepting states are all the $q'_j$'s. The initial state has a self loop with $\Sigma_s$. Once it reads an $a_j$ it moves to the respective $q_j$. From every $q_i$ or $q'_i$ when reading $a_k$ for $k > i$, we move to $q_k$, thus `remembering' the maximal $a_j$ seen. We advance from $q_j$ to $q'_j$ only when reading $s_j$, thus reaching an accepting state only if $s_j$ appears for the current maximal $a_j$ witnessed. As this is a duo-normalized FDFA, even if the required $s_j$ appeared before the maximal $a_j$, the progress automaton will reach $q'_j$ when reading a duo-normalization of the period. All other transitions are self loops (that is, reading $a_i$ for $i \leq k$ from $q_k$ and $q'_k$, or $\Sigma_s \setminus \{s_j\}$ from $q_j$, or any $\Sigma_s$ from $q'_j$).

\end{proof}

\end{document}